\documentclass[twoside,final,11pt]{entics} 
\usepackage{enticsmacro}
\usepackage{graphicx}
\usepackage[all]{xy}
\usepackage[mathscr]{eucal}
\usepackage{amssymb}
\usepackage{amsmath}
\usepackage{cleveref}
\usepackage{xparse}
\usepackage{ccicons}
\usepackage{bussproofs}

\usepackage{cmll}
\usepackage{tikz}
\usepackage{tikz-cd}
\usetikzlibrary{matrix}
\usetikzlibrary{arrows}
\usetikzlibrary{arrows.meta}
\usetikzlibrary{decorations.pathmorphing,shapes}
\usetikzlibrary{decorations.markings}

\newcommand{\commentout}[1]{}

\tikzset{
    dist/.style={->, decoration={"\mid" marking}},
    span/.style={decoration={"\mid" marking}},
    acyclic/.style={decoration={"\sim" above}},
    acyclic'/.style={decoration={"\sim" below}},
    cof/.style={>->},
    acycof/.style={cof, acyclic},
    acycof'/.style={cof, acyclic'},
    fibrant/.style={->>},
    acyfib/.style={fibrant, acyclic},
    acyfib'/.style={fibrant, acyclic'},
    spanmaps/.style={decoration={"\circ" marking}},
    sepfib/.style={fibrant, spanmaps},
}
\tikzcdset{large/.style={column sep=2cm}}
\NewDocumentCommand{\morphism}{somO{}m}{%
    \IfNoValueF{#2}{#2:}%
    \begin{tikzcd}
        [ ampersand replacement=\&, \IfBooleanT{#1}{cramped} ]
        #3\arrow[r,#4] \& #5
    \end{tikzcd}%
}

\newcommand{\Xmonobjc}{\gamma}
\newcommand{\Xmonobjcp}{\gamma'}

\newcommand{\tuple}[1]{\langle #1\rangle}

\newcommand{\Iso}[1]{\mathbf{Iso}(#1)}

\newcommand{\Nerve}[1]{N_A}

\newcommand{\identity}[1]{\mathbf{id}_{#1}}

\newcommand{\ty}{a}

\newcommand{\id}{id}
\makeatletter
\newcommand{\colim@}[2]{\vtop{\m@th\ialign{##\cr
    \hfil$#1\operator@font lim$\hfil\cr
    \noalign{\nointerlineskip\kern1.5\ex@}#2\cr
    \noalign{\nointerlineskip\kern-\ex@}\cr}}}
\newcommand{\colim}{%
  \mathop{\mathpalette\colim@{\rightarrowfill@\textstyle}}\nmlimits@
}

\newcommand{\Sym}[1]{\mathbf{Sym}#1}
\newcommand{\SymPlus}[1]{\Sym_{!} #1}
\newcommand{\SymProd}[1]{\Sym_{?} #1}
\newcommand{\SymBang}[1]{\Sym_{!} #1}
\newcommand{\SymWhyNot}[1]{\Sym_{?} #1}
\newcommand{\FamAface}[1]{\mathbf{Fam}_{+} #1}
\newcommand{\FamBface}[1]{\mathbf{Fam}_{\times} #1}
\newcommand{\whynotword}[1]{\langle#1\rangle}
\newcommand{\whynotmorph}[2]{\whynotword{#1;#2}}
\newcommand{\genwhynotword}[2]{\whynotword{#1_{1},\dotsc,#1_{#2}}}
\newcommand{\genwhynotmorph}[3]{\whynotmorph{#1}{#2_{1},\dotsc,#2_{#3}}}

\newcommand{\bangword}[1]{[#1]}
\newcommand{\bangmorph}[2]{\bangword{#2;#1}}
\newcommand{\genbangword}[2]{\bangword{#1_{1},\dotsc,#1_{#2}}}
\newcommand{\genbangmorph}[3]{\bangmorph{#1}{#2_{1},\dotsc,#2_{#3}}}
\newcommand{\permutationgroup}[1]{\mathfrak S_{#1}}

\definecolor{red}{rgb}{1,0,0}
\definecolor{green}{rgb}{0,.5,0}
\definecolor{antigreen}{rgb}{.8,0,.8}
\definecolor{blue}{rgb}{0,0,1}
\definecolor{paleblue}{rgb}{.2,.2,.6}
\definecolor{littleblue}{rgb}{.2,.2,.6}
\definecolor{darkblue}{rgb}{.3,.3,1}
\definecolor{graille}{rgb}{.01,.01,.00}
\definecolor{yellow}{rgb}{.5,.5,.2}
\definecolor{cyan}{rgb}{0.75,1,1}
\definecolor{lightgreen}{rgb}{.139,.69,.19}

\newcommand{\ie}{\emph{i.e.}}

\newcommand*{\seqdots}[3]{\whynotword{#1_{#2}, \dots, #1_{#3} } }
\newcommand{\tyl}{\vec{a}}
\newcommand{\tylother}{\vec{b}}
\newcommand{\tysf}{\vec{f}}
\newcommand{\Shuffle}{\mathbf{Shuffle}}
\newcommand{\contra}[2]{\mathbf{contr}(#1,#2)}

\newcommand{\familyBface}[3]{\langle #1 \rangle_{#2 \in #3}}

\newcommand{\push}[1]{\textsf{push}_{#1}}

\newcommand{\spanR}{\mathcal{R}}
\newcommand{\spanS}{\mathcal{S}}
\newcommand{\spanT}{\mathcal{T}}
\newcommand{\LeftAdjoint}[1]{L_{#1}}
\newcommand{\RightAdjoint}[1]{R_{#1}}

\newcommand{\distrib}{\mathbf{dist}}

\newcommand{\Set}{\mathbf{Set}}
\newcommand{\Cat}{\mathbf{Cat}}

\newcommand{\CatId}{\Cat^{=}}
\newcommand{\CatNat}{\Cat^{\Rightarrow}}
\newcommand{\CatEq}{\Cat^{\simeq}}
\newcommand{\sourcefunctor}{\mathbf{source}}
\newcommand{\targetfunctor}{\mathbf{target}}

\newcommand{\DistBin}[2]{\mathbf{Dist}(#1,#2)}
\newcommand{\Dist}{\mathbf{Dist}}
\newcommand{\Acategory}{A}
\newcommand{\Bcategory}{B}
\newcommand{\Ccategory}{C}
\newcommand{\Rcategory}{R}
\newcommand{\Scategory}{S}
\newcommand{\Tcategory}{T}
\newcommand{\Xcategory}{X}

\newcommand{\relsem}[1]{[\![ #1 ]\!]}
\newcommand{\sem}[1]{[\![ #1 ]\!]}

\newcommand{\SomeProof}[2]{#1\,\vdash\,#2}
\newcommand{\Xtensor}{\otimes}

\newcommand{\Rel}{\mathbf{Rel}}
\newcommand{\BiCatRel}{\mathbf{\mathscr{W}}}

\newcommand{\SpanCat}{\mathbf{Span}}
\newcommand{\SpanCatBin}[2]{\mathbf{Span}(#1,#2)}
\newcommand{\HoSpanCat}{\mathbf{HoSpan}}

\newcommand{\FibSpanCat}{\mathbf{SFibSpan}}
\newcommand{\true}{\mathbf{true}}
\newcommand{\false}{\mathbf{false}}
\newcommand{\relsource}{\mathbf{source}}
\newcommand{\reltarget}{\mathbf{target}}
\newcommand{\IntGpd}{\mathbb{J}}

\newcommand{\Rsourcefunctor}{\mathbf{source}_R}
\newcommand{\Rtargetfunctor}{\mathbf{target}_R}
\newcommand{\Ssourcefunctor}{\mathbf{source}_S}
\newcommand{\Stargetfunctor}{\mathbf{target}_S}
\newcommand{\Tsourcefunctor}{\mathbf{source}_T}
\newcommand{\Ttargetfunctor}{\mathbf{target}_T}

\newcommand{\concat}{\mathbf{concat}}
\newcommand{\deshuffle}{\mathbf{deshuffle}}

\newcommand{\zerocategory}{\mathbf{0}}
\newcommand{\onecategory}{\mathbf{1}}

\newcommand{\one}{1}

\newcommand{\fstproj}{\mathtt{fst}}
\newcommand{\sndproj}{\mathtt{snd}}

\newcommand{\differential}{}

\newcommand\monlift[1]{#1^\dagger}
\newcommand\sfiblift[1]{#1^\ddagger}

\def\conf{MFPS 2026} 	
\volume{6}			
\def\volu{6}			
\def\papno{11}			
\def\mydoi{19635}
\def\lastname{Melliès, Olimpieri, Vaux Auclair} 
\def\runtitle{An explicit shuffle construction\ldots} 
\def\copynames{P.-A. Melliès, F. Olimpieri, L. Vaux Auclair}    

\begin{document}

\begin{frontmatter}
\title{An Explicit Shuffle Construction\\ of the Homotopy Span Model of \differential Linear Logic} 

  \author{Paul-André Melliès\thanksref{a}\thanksref{afunding}}	
  \author{Federico Olimpieri\thanksref{b}\thanksref{bfunding}}		
  \author{Lionel Vaux Auclair\thanksref{c}\thanksref{cfunding}}		
  \address[a]{CNRS, Université Paris Cité, Inria, Paris, France}
  \address[b]{Aix Marseille Univ, CNRS, LIS, Marseille, France}
  \address[c]{Aix Marseille Univ, CNRS, I2M, Marseille, France}
  \thanks[afunding]{This work has received funding from the European Research Council
under the European Union’s Horizon 2020 research and innovation
programme (Synergy Project Malinca, ERC Grant Agreement No 670624).}
\thanks[bfunding]{This work received support from the French government under the France 2030 investment plan, as a part of the Initiative d'Excellence d'Aix-Marseille Université-AMIDEX, AMX-CPJ-2024-09.}
  \thanks[cfunding]{Partially funded by French ANR projects PPS (ANR-19-CE48-0014), Lambda-Comb (ANR-21-CE48-0017), and ReCiProg (ANR-21-CE48-0019).}

\begin{abstract}
We give a direct and explicit construction of the homotopy span model of \differential linear logic 
originally formulated by Melli{\`e}s using abstract homotopy theory in the Quillen category~$\mathbf{Cat}$
of small categories equipped with its natural model structure.
%
%
%
Following the principles of homotopy theory, the formulas and proofs of linear logic
are interpreted in the model as small categories up to categorical equivalence,
the notion of weak homotopy equivalence underlying the natural Quillen model structure on~$\Cat$.
%
%
%
Here, we explain how the original interpretation of proofs as \emph{fibrant spans}
can be replaced by a more liberal interpretation of proofs as \emph{separately fibrant spans}.
This relaxation from fibrant spans to separately fibrant spans enables us
to give a simple and intuitive description of the homotopy span model,
where the axiom and cut links are interpreted as identity spans,
and where the symmetries of the exponential modality are captured
by a shuffle structure on the interpretation of proofs.
%
%
\end{abstract}
\end{frontmatter}

\section{Introduction}
The relational semantics of linear logic \cite{Girard87,ehrhard-completeness,ehr:fin}
has become today a fundamental tool to investigate and to analyze
the structure of higher-order computations.
It provides at the same time an abstract mathematical framework
which played a central role in the discovery of differential linear logic
\cite{ehr:fin,EhrhardRegnier06,Ehrhard18}, 
and a concrete computational interpretation through its strong connection 
to intersection types~\cite{deCarvalho18,BucciarelliEhrhard00,BucciarelliEhrhard01,Terui12}
and higher-order model checking~\cite{KobayashiOng11,GrelloisMellies15}.

About twenty years ago, Fiore, Gambino, Hyland and Winskel
introduced a categorified version of the relational semantics
based on distributors, motivated by an unexpected link
with Joyal's theory of combinatorial species~\cite{FioreGambinoHylandWinskel07}.
This step of categorification is needed in order to obtain a concrete interpretation of linear logic,
where the elements of the denotation of a proof coincide with the (equivalence classes of)
derivation trees typing the proof in a specific intersection type system~\cite{Olimpieri21}.

The model of generalized species has been a source of inspiration for a new generation
of categorified models based on distributors~\cite{joyal:op,Galal20,FioreGalalPaquet24}.
%
%
An alternative way to categorify relations between sets
is to refine them into functorial spans between categories.
Melli{\`e}s explained how to construct a functorial span model 
of linear logic~\cite{Mellies19lics} based on template games,
where functorial spans are considered up to homotopy in the natural Quillen model 
structure.
Following the homotopy span model, other approaches to functorial spans
based on biorthogonality~\cite{CF2023} or $\infty$-categories~\cite{HaringtonMimram25} have been developed.
%
%
In this paper, we will give a more explicit presentation of the homotopy span model,
exploiting the key notion of shuffle permutation on the letters of a finite word.
We start by recalling the homotopy span model and its relationship 
with the generalized species interpretation of linear logic.

\paragraph{Two alternative categorifications of relations}
Given two sets~$A$ and~$B$, 
every binary relation $R\subseteq A\times B$ 
may be equivalently described
as a function $A\times B\to \mathbf{2}$ where $\mathbf{2}=\{\true,\false\}$ 
denotes the set of booleans, 
or as a span~$A\leftarrow S\rightarrow B$ consisting
of a set~$S$ and of two jointly injective functions $\relsource:S\to A$ and $\reltarget:S\to B$.
%
%
Each point of view induces a specific
way to \emph{categorify} the notion of binary relation,
and to refine the category~$\Rel$ of sets and relations
into a bicategory~$\BiCatRel$ 
of \emph{categorified relations} 
between small categories~$\Acategory$ and~$\Bcategory$.

%
In the first approach, one defines the bicategory~$\BiCatRel=\Dist$ of 
small categories and \emph{distributors},
where a distributor
$
\begin{tikzcd}[column sep = 1em]
M \, : \, \Acategory\arrow[rr,dist] && \Bcategory
\end{tikzcd}
$
is defined as a set-valued functor
$M : \Acategory\times\Bcategory^{\hspace{.05em}op}\to\Set$,
composed by a coend formula,
and where a 2-cell $\theta:M\Rightarrow N$ between distributors 
$
\begin{tikzcd}[column sep = 1em]
M , N \, : \, \Acategory\arrow[rr,dist] && \Bcategory
\end{tikzcd}
$
is defined as a natural transformation
between the underlying functors.

In the second approach, one defines the bicategory~$\BiCatRel=\SpanCat$ 
of small categories and \emph{spans of functors}
where a span
$
\begin{tikzcd}[column sep = 1em]
\spanS \, : \, \Acategory\arrow[rr,span] && \Bcategory
\end{tikzcd}
$
is defined as a triple $\spanS=(\Scategory,\sourcefunctor,\targetfunctor)$
consisting of a small category~$\Scategory$ called the \emph{support category}
of~$\spanS$, and of a pair of functors
\begin{equation}\label{equation/functorial-span}
\begin{tikzcd}[column sep = 3em]
\Acategory && \Scategory   \arrow[ll,"{\sourcefunctor}"{swap}] \arrow[rr,"{\targetfunctor}"] && \Bcategory
\end{tikzcd}
\end{equation}
composed by computing pullbacks in $\Cat$,
and where a 2-cell $\theta:\spanS\Rightarrow\spanT$ 
(called a \emph{simulation} 2-cell)
between spans 
$
\begin{tikzcd}[column sep = 1em]
\spanS, \spanT \, : \, \Acategory\arrow[rr,span] && \Bcategory
\end{tikzcd}
$
is defined as a functor $\theta:\Scategory\to\Tcategory$
between the underlying support categories,
making the diagram commute:
\begin{equation}\label{equation/2-cell-between-spans}
\begin{tikzcd}[column sep = 3em, row sep = .2em]
&& 
{\Scategory}
\arrow[rrd,"{\Stargetfunctor}"]
\arrow[lld,"{\Ssourcefunctor}"{swap}]
\arrow[dd,"{\theta}"]
\\
{\Acategory}
&&&&
{\Bcategory\,.}
\\
&&
{\Tcategory}
\arrow[rru,"{\Ttargetfunctor}"{swap}]
\arrow[llu,"{\Tsourcefunctor}"]
\end{tikzcd}
\end{equation}
Note that $\SpanCat$ is the usual span construction 
of a bicategory on the category $\Cat$ with pullbacks.
Interestingly, both bicategories~$\Dist$ and~$\SpanCat$ 
are $\ast$-autonomous (in fact compact closed) 
with tensor product $\Acategory,\Bcategory\mapsto\Acategory\otimes\Bcategory$
defined as categorical product~$\Acategory,\Bcategory\mapsto\Acategory\times\Bcategory$,
as well as cartesian and cocartesian with sum $\Acategory,\Bcategory\mapsto\Acategory\oplus\Bcategory$
and product $\Acategory,\Bcategory\mapsto\Acategory\&\Bcategory$
both defined as categorical sum $\Acategory,\Bcategory\mapsto\Acategory+\Bcategory$.
For that reason, $\Dist$ and~$\SpanCat$ define models of the multiplicative additive 
fragment MALL of linear logic.
%
%

\paragraph{Similitudes and differences between $\Dist$ and~$\SpanCat$.}
Before reviewing how the bicategories
$\Dist$ and~$\SpanCat$
were equipped in~\cite{FioreGambinoHylandWinskel07,Mellies19lics}
with an interpretation of the exponential modality $A\mapsto {!}A$
to obtain a model of \differential linear logic, we 
find useful to 
mention that both categorifications
share an important formal property:
namely,
each functor $F:\Acategory\to\Bcategory$
induces an adjunction
\[
\begin{tikzcd}[column sep=.8em]
\LeftAdjoint{F}
\quad : 
&
\Acategory
\hspace{.4em}
\arrow[rrrr,span,yshift=.1em,bend left = 25]
&& \bot &&
\hspace{.4em}
\Bcategory
\arrow[llll,span,yshift=0pt,bend left = 25]
&
 : \quad
\RightAdjoint{F}
\end{tikzcd}
\]
in the bicategories~$\BiCatRel=\Dist$ and~$\BiCatRel=\SpanCat$.
In $\Dist$, the adjunction is defined as follows:
\[
\begin{array}{cc}
\begin{tikzcd}[column sep = 1em]
\LeftAdjoint{F}
\hspace{.5em} = \hspace{.5em}
(a,b) \hspace{.2em} \mapsto \hspace{.2em} \Bcategory(b,Fa) 
\hspace{.5em} : \hspace{.5em}
\Acategory\times\Bcategory^{op}
\arrow[rr]
&&
\Set
\end{tikzcd}
\\
\begin{tikzcd}[column sep = 1em]
\RightAdjoint{F}
\hspace{.5em} = \hspace{.5em}
(b,a) \hspace{.2em} \mapsto \hspace{.2em} \Bcategory(Fa,b) 
\hspace{.5em} : \hspace{.5em}
\Bcategory\times\Acategory^{op}
\arrow[rr]
&&
\Set
\end{tikzcd}
\end{array}
\]
while the adjunction~$\LeftAdjoint{F}\dashv\RightAdjoint{F}$
of spans in $\SpanCat$ is defined as follows
\[
\begin{tikzcd}[column sep =1.5em]
\LeftAdjoint{F}
\hspace{.5em} = \hspace{.5em}
\Acategory && 
\Acategory
\arrow[ll,"{\identity{\Acategory}}"{swap}]
\arrow[rr,"{F}"]
&&
\Bcategory
\end{tikzcd}
\hspace{3em}
\begin{tikzcd}[column sep = 1.5em]
\RightAdjoint{F}
\hspace{.5em} = \hspace{.5em}
\Bcategory && 
\Acategory
\arrow[ll,"{F}"{swap}]
\arrow[rr,"{\identity{\Acategory}}"]
&&
\Acategory\,.
\end{tikzcd}
\]
%
Note that if we consider the category $\Cat$ as the 2-category $\CatId$ with
identity $2$-cells only, then 
the construction $F\mapsto \LeftAdjoint{F}$ is functorial in both cases,
in the sense that it defines a weak 2-functor between bicategories (also called
a pseudofunctor) $\LeftAdjoint{}:\CatId\to\BiCatRel$
for both $\Dist$  and $\SpanCat$.
On the other hand, one key difference between~$\Dist$ and~$\SpanCat$ is
that if, instead of $\CatId$, we consider $\Cat$ as the 2-category $\CatNat$
with all natural transformations as $2$-cells, then $\LeftAdjoint{}:\CatId\to\BiCatRel$ extends to
a weak 2-functor $\LeftAdjoint{}:\CatNat\to\BiCatRel$ in the case
$\BiCatRel=\Dist$ but not in the case $\BiCatRel=\SpanCat$.
Indeed, a $2$-cell from $\LeftAdjoint{F}$ to $\LeftAdjoint{G}$ in $\SpanCat$
is given by an endofunctor $\theta:A\to A$ making the
diagram~\eqref{equation/2-cell-between-spans} commute
with both $\Ssourcefunctor$ and $\Tsourcefunctor$ being identity functors $\identity{\Acategory}$:
the only possible choice for $\theta$ is also $\identity{\Acategory}$, which forces $F=G$.
%

From now on, whenever we consider $\Cat$ as a 2-category (and thus bicategory),
we do mean $\CatNat$, unless specified otherwise.



%
%
%

\paragraph{The free symmetric monoidal category construction.}
As we will see below, the interpretation of the exponential modality $A\mapsto {!}A$
in the bicategories~$\BiCatRel=\Dist$ and~$\BiCatRel=\SpanCat$
is derived in~\cite{FioreGambinoHylandWinskel07,Mellies19lics}
from the 2-monad~$\Sym:\Cat\to\Cat$ which transports
every small category~$\Acategory$ to the symmetric strict monoidal 
category~$\Sym\,\Acategory$ freely generated by~$\Acategory$.
In both cases, $\Sym$ categorifies the finite multiset monad
$\mathbf{MSet}_{\mathrm{fin}}:\Set\to\Set$ from which one derives the
exponential modality $A\mapsto {!}A$ in the relational semantics of
\differential linear logic.

We recall that the category~$\Sym\,\Acategory$
%
has its objects defined as finite words $\genwhynotword{a}{n}$ 
of objects of $\Acategory$, and its morphisms 
$\genwhynotword{a}{n}\to\genwhynotword{a'}{n}$
defined as tuples $\genwhynotmorph{\sigma}{f}{n}$ 
consisting of a permutation $\sigma\in\permutationgroup{n}$
together with a finite word~$\genwhynotword{f}{n}$
of morphisms $f_i:a_{\sigma(i)}\to a'_i$ in the category~$\Acategory$ ;
there are no morphisms between words
of different lengths.
The construction defines a 2-monad~$\Sym:\Cat\to\Cat$ 
equipped with cartesian natural 2-transformations:
the multiplication $\mu:\Sym\circ\Sym\Rightarrow\Sym$ given by the
concatenation of words of words, and unit $\eta:Id\Rightarrow\Sym$ given by
the singleton word.

Moreover, every pair of small categories~$\Acategory$ and $\Bcategory$
induces a pair of functors
\begin{equation}\label{equation/seely-equivalence-1}
\begin{tikzcd}[column sep=3.4em]
\Sym\,(\Acategory+\Bcategory)
\,\,
\arrow[rr,"{\deshuffle_{\Acategory,\Bcategory}}",yshift=.35em]
&&
\,\,
\Sym\,\Acategory\,\times\,\Sym\,\Bcategory
\arrow[ll,"{\concat_{\Acategory,\Bcategory}}",yshift=-.35em]
\end{tikzcd}
\end{equation}
defining an equivalence of categories, as noted in \cite[Lemma 5.2.1]{FioreGambinoHylandWinskel07}.
These functors, which can be derived from the universal
property of $\Sym$, are explicitly defined as follows:
\[
\begin{tikzcd}[column sep=.5em,row sep=0em]
\concat_{\Acategory,\Bcategory} & = & (\genwhynotword{a}{m},\genwhynotword{b}{n})
& \arrow[rrr,|->] &&& \hspace{1.5em} \whynotword{a_1,\dotsc,a_m,b_1,\dotsc,b_n} \hspace{1.5em}
\\
\deshuffle_{\Acategory,\Bcategory} & = & \genwhynotword{c}{m+n}
& \arrow[rrr,|->] &&& \quad  (\genwhynotword{a}{m},\genwhynotword{b}{n})
\end{tikzcd}
\]
where
the finite words $\genwhynotword{a}{m}$ and $\genwhynotword{b}{n}$ 
are uniquely characterized 
by the fact that there exist jointly surjective functions
$f:\{1,\dots,m\}\to\{1,\dots,m+n\}$ and $g:\{1,\dots,n\}\to\{1,\dots,m+n\}$
such that $c_{f(i)}=a_{i}$ for all $1\leq i\leq m$
and $c_{g(j)}=b_j$ for all $1\leq j\leq n$.
The families $\concat$ and $\deshuffle$ define
natural transformations between the 2-functors~$\Cat\times\Cat\to\Cat$
defined as $\Acategory,\Bcategory\mapsto\Sym\,(\Acategory+\Bcategory)$
and $\Acategory,\Bcategory\mapsto\Sym\,\Acategory\times\Sym\,\Bcategory$.
Both natural 2-transformations are moreover cartesian,
in the sense that the commuting squares for naturality 
are pullback squares.
The 2-monad $\Sym:\Cat\to\Cat$ also comes with functors
\begin{equation}\label{equation/seely-equivalence-2}
\begin{tikzcd}[column sep=3.4em]
\Sym\,\zerocategory
\,\,
\arrow[rr,"{\deshuffle_{\zerocategory}}",yshift=.35em]
&&
\,\,
\onecategory
\arrow[ll,"{\concat_{\zerocategory}}",yshift=-.35em]
\end{tikzcd}
\end{equation}
defining an equivalence of categories,
where~$\zerocategory$ is the empty category, and $\onecategory$ is the singleton category.
The family of functors $\concat$ equips the 2-functor~$\Sym:\Cat\to\Cat$
with the structure of a symmetric monoidal 2-functor from $(\Cat,+,\zerocategory)$
to $(\Cat,\times,\onecategory)$, whose monoidal coercions, in both directions,
are categorical equivalences.
For these reasons the 2-monad $\Sym$ may be called a cartesian symmetric
monoidal 2-monad.

From now on, we will call \eqref{equation/seely-equivalence-1} 
and \eqref{equation/seely-equivalence-2} \emph{Seely equivalences} 
since they induce what plays the rôle of \emph{Seely isomorphisms}
in the categorified interpretation of \differential linear logic discussed
below.

\paragraph{The distributor model of \differential linear logic.}
%
Fiore, Gambino, Hyland and Winskel
explained in their seminal and celebrated paper
on generalised species~\cite{FioreGambinoHylandWinskel07}
how to use the free symmetric monoidal 
category 2-monad $\Sym:\Cat\to\Cat$
in order to turn $\Dist$ into a 2-dimensional model of \differential linear logic.
The construction relies on two main ingredients.

The first ingredient is the existence of a distributivity law 
between~$\Sym$ and the presheaf (relative) monad on~$\Cat$ \cite{fiore:rel},
which enables one to extend the 2-monad $\Sym:\Cat\to\Cat$
to a weak 2-monad~$\SymWhyNot:\Dist\to\Dist$ --
see the discussion in~\cite{Hyland10}.
The weak 2-monad is then transformed by duality into a weak 2-comonad
$\SymBang:\Dist\to\Dist$
obtained by precomposing and postcomposing $\SymWhyNot$ 
with the 2-functor $\mathbf{op}:\Dist\to\Dist^{op(1)}$
which sends every small category~$\Acategory$ 
to its opposite category~$\Acategory^{\hspace{.05em}op}$.
Here, we write $\Dist^{op(1)}$ for the bicategory~$\Dist$
where the orientation of 1-cells (but not of 2-cells) has been reversed.

The second ingredient, less known but just as important as the first one, is
the fact that, as discussed above, the weak 2-functor $\LeftAdjoint{}:\CatId\to\Dist$
does extend to a weak 2-functor $\LeftAdjoint{}:\CatNat\to\Dist$.
This ensures that $\LeftAdjoint{}$ transports
the Seely equivalence~\eqref{equation/seely-equivalence-1} 
and~\eqref{equation/seely-equivalence-2} from the 2-category~$\Cat$
to the corresponding Seely equivalence in the bicategory~$\Dist$.
%

%
%


%

\paragraph{An apparent obstruction to the construction of the span model of linear logic.}
Somewhat surprisingly, one bumps against a serious obstruction
when one tries to obtain a model of linear logic based on spans,
by extending the 2-monad $\Sym:\Cat\to\Cat$ to the
bicategory~$\SpanCat$ instead of~$\Dist$.
Indeed, as mentioned above and in contrast to what happens with~$\Dist$,
the weak 2-functor $\LeftAdjoint{}:\CatId\to\SpanCat$ does not extend to
$\CatNat$. Hence, it does not transport the Seely
equivalence~\eqref{equation/seely-equivalence-1} living in the
2-category~$\CatNat$
to a Seely equivalence living in the bicategory~$\SpanCat$.
On the other hand, we are not far from a linear logic model, since $\Sym$ lifts
to a weak 2-monad on $\SpanCat$ as a cartesian 1-monad on $\Cat $
\cite{burroni1971tcat,paquet2023strong};
$\Sym$ is moreover both lax and oplax monoidal on $\Cat $, as supported by the
Seely equivalences \eqref{equation/seely-equivalence-1}.
For this reason, $\Sym $ may be called a bilax monoidal (2-)monad on $\Cat $
\cite{AguiarMahajan10},
and the whole structure lifts to $\SpanCat $ because
\eqref{equation/seely-equivalence-1} is cartesian.
However, the fact that $L $ does not extend to the 2-dimensional structure of
$\CatNat$ implies in particular that the equivalence
\eqref{equation/seely-equivalence-1} in $\CatNat $ is \emph{not} transported to
an equivalence in the bicategory $\SpanCat$.


Melli\`es~\cite{Mellies19lics} explained
how to use abstract homotopy theory in order to resolve the obstruction
and to define a homotopy variant~$\HoSpanCat$ of the original bicategory~$\SpanCat$
suitable for the interpretation of \differential linear logic.
%
The philosophy of abstract homotopy theory is that the ``good'' homotopic constructions
performed in a Quillen model category
should not depend on the choice of objects up to homotopy.
%
Given an object, there is often a large number of alternative choices
of objects equivalent up to homotopy, and a good practice is thus 
to pick a specific object defined in a principled way.

In his original construction, Melli{\`e}s~\cite{Mellies19lics} 
interpreted proofs using fibrant objects and fibrant replacements in the hom-category $\SpanCat(\Acategory,\Bcategory)$.
%
Our goal in the present work is to give a more direct and explicit
construction of this 2-dimensional model, 
where the fibrant spans are replaced by separately fibrant spans
carefully chosen at each step of the interpretation.
This construction highlights the combinatorial nature of the interpretation,
and reveals the existence of a factorization system on shuffles 
in the symmetric monoidal category $\Sym(A)$, seen as an operad.

\paragraph{Plan of the paper.}
After this long but necessary introduction on the bicategorical semantics of
linear logic, we will expose in more detail the homotopy span model recently
formulated by Melli\`es~\cite{Mellies19lics} in \cref{section/hospan}.
In particular, we explain how it takes advantage of homotopy ideas in order to
transport the Seely equivalences in $\CatNat$ to Seely equivalences in the
homotopy variant $\HoSpanCat$, constructed from $\SpanCat$ using the natural
Quillen model structure on~$\Cat$.

In \cref{section/LL} and onwards, we describe how the formulae and proofs
of linear logic are interpreted in our new and more explicit formulation of the
model, based on shuffles.
The interpretation relies in particular on a factorization property of the operad $\Sym(A)$,
of independent interest, which we choose to present in \cref{section/factorization}.
One main contribution is to show how this factorization property
allows to give a simple and concrete separately fibrant interpretation
of promotion: this is done in \cref{section/promotion-rule}.
We then conclude in \cref{section/conclusion}.

\section{The homotopy span model of \differential linear logic}
\label{section/hospan}

The key observation made in~\cite{Mellies19lics} is that a natural isomorphism
$\varphi:F\Rightarrow G$ between functors $F,G:\Acategory\to\Bcategory$ in
$\Cat$
is \emph{not} transported to an invertible 2-cell
${\LeftAdjoint{\varphi}}:{\LeftAdjoint{F}}\Rightarrow{\LeftAdjoint{G}}$
in the bicategory~$\SpanCat$,
%
but to a \emph{cospan} of 2-cells from $\LeftAdjoint{F}$ and $\LeftAdjoint{G}$ 
to the span $\widetilde{L}_{\varphi}$:
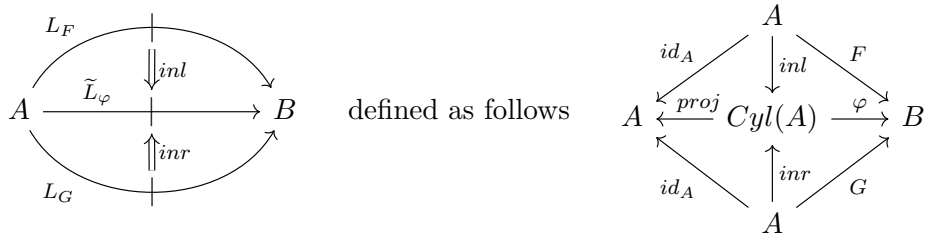
\begin{center}
\begin{tikzcd}[row sep=0pt,column sep=2.5em]
\Acategory
\arrow[rrr,span,"{L_F}"{near start},""{swap,name=sourceup}, bend left = 60]
\arrow[rrr,span,"{\widetilde{L}_{\varphi}}"{near start},""{name=targetup},""{name=targetdown,swap}]
\arrow[rrr,span,"{L_G}"{swap,near start},""{name=sourcedown}, bend right = 60]
&&&
\Bcategory
\arrow[Rightarrow,swap,from=sourceup,to=targetup,shorten <=3pt, shorten >=3pt,"{inl}"{swap}]
\arrow[Rightarrow,swap,from=sourcedown,to=targetdown,shorten <=3pt, shorten >=3pt,"{inr}"]
\end{tikzcd}
\quad defined as follows \quad
\begin{tikzcd}[row sep=1em,column sep=1em]
&&
A\arrow[dd,"{inl}"]\arrow[rrdd,"F"]\arrow[lldd,"{id_A}"{swap}]
\\
\\
A
&&
Cyl(A)\arrow[ll,"proj"{swap,near start}]\arrow[rr,"{\varphi}"]
&&
B
\\
\\
&&
A\arrow[uu,"{inr}"{swap}]\arrow[rruu,"G"{swap}]\arrow[lluu,"{id_A}"]
\end{tikzcd}
\end{center} 
Here, the span $\widetilde{L}_{\varphi}$ is defined as $(Cyl(A),proj,\varphi)$
where $Cyl(A)$ denotes the \emph{cylinder category}
defined as the cartesian product $Cyl(\Acategory)\, = \,\IntGpd\times \Acategory$
of the small category~$\Acategory$ with the \emph{interval groupoid} $\IntGpd$ 
with two objects~$0$ and~$1$
and a unique isomorphism $0\to 1$ between them.
The functors $inl$, $inr$ and $proj$ are defined as
$inl (a) := (0,a)$, $inr (a) := (1,a)$ and $proj(i,a):=a$ for $i\in\{0,1\}$,
while the functor $\varphi:\IntGpd\times A\to B$ provides
an alternative description of the natural isomorphism~$\varphi$
with $\varphi(0,a)=Fa$, $\varphi(1,a)=Ga$, $\varphi(0\to 1,a)=\varphi_a$.

An additional observation made in~\cite{Mellies19lics}
is that the functors $inl,inr:\Acategory\to Cyl(\Acategory)$ 
are fully faithful and essentially surjective, and thus equivalences of categories.
This strongly suggests to \emph{localize} (or formally invert)
in the hom-category~$\SpanCat(\Acategory,\Bcategory)$
every morphism given by a simulation 2-cell~\eqref{equation/2-cell-between-spans}
whose underlying functor $\theta:\Scategory\to\Tcategory$ is an equivalence of categories.
%
%
This localization can be neatly performed using the tools of abstract
homotopy theory, and in particular the notion of Quillen model category.
Recall \cite{JoyalTierney1991,Rezk1996} that the category~$\Cat$ comes equipped the natural Quillen
structure $(\mathcal{W}, \mathcal{C}, \mathcal{F})$ defined as follows:
\begin{itemize}
\item the class $\mathcal{W}$ of \emph{weak equivalence} is the class of equivalences of categories,
\item the class $\mathcal{C}$ of \emph{cofibrations} is the class of functors injective on objects,
\item the class $\mathcal{F}$ of \emph{fibrations} is the class of isofibrations --
    recall that a functor $F:\Acategory\to\Bcategory$ is an isofibration if for
    any object $a\in \Acategory$ and any isomorphism $g$ with target
    $F(a)$ in $\Bcategory$, there exists an isomorphism $f$
    with target $a$ in $\Acategory$ such that $g=F(f)$.
\end{itemize}
From this follows \cite{Hirschhorn03} that every hom-category $\SpanCat(\Acategory,\Bcategory)=\Cat/\Acategory\times\Bcategory$ 
comes equipped as a slice category with a Quillen model structure 
$(\mathcal{W}_{\Acategory,\Bcategory}, \mathcal{C}_{\Acategory,\Bcategory}, \mathcal{F}_{\Acategory,\Bcategory})$:
\begin{itemize}
\item a weak equivalence
in $\mathcal{W}_{\Acategory,\Bcategory}$
is a simulation 2-cell~$\theta:\spanS\Rightarrow\spanT$ whose underlying functor~$\theta:\Scategory\to\Tcategory$
is an element of $\mathcal{W}$, that is, an equivalence of categories;
\item a cofibration in $\mathcal{C}_{\Acategory,\Bcategory}$ is a simulation 2-cell~$\theta:\spanS\Rightarrow\spanT$
whose underlying functor~$\theta:\Scategory\to\Tcategory$ is an element of~$\mathcal{C}$, that is,
a functor injective on objects; 
\item a fibration in $\mathcal{F}_{\Acategory,\Bcategory}$
is a simulation 2-cell~$\theta:\spanS\Rightarrow\spanT$ whose underlying functor~$\theta:\Scategory\to\Tcategory$ is an element of~$\mathcal{F}$, that is,
a isofibration.
\end{itemize}
This defines what we like to call the natural model structure on $\SpanCat(\Acategory,\Bcategory)=\Cat/A\times B$.
The category obtained by localizing in $\SpanCat(\Acategory,\Bcategory)$
all weak equivalences in $\mathcal{W}_{\Acategory,\Bcategory}$ is then
the homotopy category
$\HoSpanCat(\Acategory,\Bcategory)=\SpanCat(\Acategory,\Bcategory)[\mathcal{W}_{\Acategory,\Bcategory}^{-1}]$
associated to this natural Quillen model structure.
%
One important point is that this category $\HoSpanCat(\Acategory,\Bcategory)$
cannot be easily described by applying a purely 2-categorical approach:
note in particular that although the weak equivalences $inl,inr:\Acategory\to Cyl(\Acategory)$ 
are adjoint equivalences in the 2-category~$\Cat$, they are not adjoint equivalences
in the 2-category $\SpanCat(\Acategory,\Bcategory)$
of functorial spans over $\Acategory$ and $\Bcategory$.


\paragraph{Morphisms in a Quillen model category.} 
We follow a well-established tradition in homotopy theory
and draw every cofibration 
$\morphism*[c]{A}[cof]{B}$
and every fibration 
$\morphism*[f]{A}[fibrant]{B}$.
We also indicate with a sign $\sim$ that a morphism~$w:A\to B$ is a weak equivalence.
Typically, an \emph{acyclic cofibration}
$\morphism*[v]{A}[acycof]{B}$
is defined as an element of $\mathcal{W}\cap\mathcal{C}$,
that is, a weak equivalence which is at the same time a cofibration;
and an \emph{acyclic fibration} 
$\morphism*[w]{A}[acyfib]{B}$
is an element of $\mathcal{W}\cap\mathcal{F}$.
By way of illustration, we could rewrite 
the Seely equivalence~\eqref{equation/seely-equivalence-1}
as
\begin{equation}\label{equation/seely-equivalence-1-homotopy}
\begin{tikzcd}[column sep=2.5em]
\deshuffle_{\Acategory,\Bcategory}\,:\,
\Sym\,(\Acategory+\Bcategory)
\,\,
\arrow[rr,->>,"{\sim}",yshift=.35em]
&&
\,\,
\Sym\,\Acategory\,\times\,\Sym\,\Bcategory
\,:\,\concat_{\Acategory,\Bcategory}
\arrow[ll,>->,"{\sim}",yshift=-.35em]
\end{tikzcd}
\end{equation}
in order to stress the fact that $\concat_{\Acategory,\Bcategory}$ is an acyclic cofibration
and that $\deshuffle_{\Acategory,\Bcategory}$ is an acyclic fibration
in the natural Quillen model structure on~$\Cat$.
Indeed, we have already stated that $\concat_{\Acategory,\Bcategory}$
and $\deshuffle_{\Acategory,\Bcategory}$ form an equivalence
of categories,
and the injectivity of $\concat_{\Acategory,\Bcategory}$
on objects is obvious.
The fact that $\deshuffle_{\Acategory,\Bcategory}$ is an isofibration
can be deduced from the fact that
$\deshuffle_{\onecategory,\onecategory}$
is an isofibration for purely combinatorial reasons,
and the fact that isofibrations are preserved by pullbacks
-- here we also apply the fact that $\deshuffle$ is a cartesian
natural 2-transformation, to the specific case of the unique
functors $F:\Acategory\to\onecategory$ and $G:\Bcategory\to\onecategory$
to the terminal category.

An object~$A$ is called \emph{cofibrant} when the unique morphism~$0\to A$
from the initial object~$0$ is a cofibration, and \emph{fibrant} when the unique morphism~$A\to 1$
to the terminal object~$1$ is a fibration. 
In a Quillen model category, every morphism $f:A\to B$ can be factored (in a non necessarily unique way)
as an acyclic cofibration followed by a fibration
\begin{tikzcd}[cramped]
    f \, : \, A\arrow[r,acycof] & A' \arrow[r,fibrant] &B
\end{tikzcd},
and as a cofibration followed by an acyclic fibration
\begin{tikzcd}[cramped]
    f \, : \, A\arrow[r,cof] & B' \arrow[r,acyfib] & B
\end{tikzcd}.
When $B=1$ and $f:A\to 1$ is the unique morphism to the terminal object, 
any object~$A'$ involved in a factorization of the first form
is called a \emph{fibrant replacement} of~$A$.
Note that an object~$A$ may have several (non necessarily isomorphic) fibrant replacements.
%

\paragraph{The homotopy category~$\HoSpanCat(\Acategory,\Bcategory)$ of spans.}
One advantage of applying the ideas of abstract homotopy theory 
is that it provides a nice description of the homotopy category 
$\mathbf{Ho}\,\Ccategory=\Ccategory[\mathcal{W}^{-1}]$
obtained by localizing (= formally inverting) all the weak equivalences
of a Quillen model category~$\Ccategory$ with class $\mathcal{W}$ of weak equivalences,
class $\mathcal{C}$ of cofibrations and class $\mathcal{F}$ of fibrations.
The homotopy category $\mathbf{Ho}\,\Ccategory=\Ccategory[\mathcal{W}^{-1}]$ 
is the category of fibrant and cofibrant objects in $\Ccategory$,
whose morphisms are defined as the morphisms of the original category $\SpanCat(\Acategory,\Bcategory)$
considered modulo a specific homotopy relation, see \cite{Hovey99,Hirschhorn03} for details.
In the particular case of the homotopy category
$\HoSpanCat(\Acategory,\Bcategory):=\SpanCat(\Acategory,\Bcategory)[\mathcal{W}_{\Acategory,\Bcategory}^{-1}]$
the characterization of the fibrant and cofibrant objects of $\SpanCat(\Acategory,\Bcategory)$ goes as follows.
Every small category in 
$\Cat$
is both fibrant and cofibrant.
The situation is a bit different in the general case of the natural Quillen model structure
on $\SpanCat(\Acategory,\Bcategory)=\Cat/A\times B$.
In that case, every span $\spanS=(\Scategory,\sourcefunctor,\targetfunctor)$ is cofibrant, because the initial object $ 0_{\Acategory, \Bcategory} $ in $\SpanCat(\Acategory,\Bcategory)  $ has source $ 0 $ and thus the unique morphism  $ 0_{\Acategory, \Bcategory} \to \spanS $ is injective on objects.
On the other hand, a span is fibrant precisely when its canonical pairing functor
\begin{equation}\label{equation/pairing-functor}
\begin{tikzcd}[column sep = 2em] 
\langle\sourcefunctor,\targetfunctor\rangle \quad : \quad \Scategory \arrow[rr,->>] && \Acategory\times\Bcategory
\end{tikzcd}
\end{equation}
is an isofibration, because the terminal object $ 1_{\Acategory, \Bcategory} $ in $ \SpanCat(\Acategory,\Bcategory) $ is the identity functor $  \Acategory \times \Bcategory \to \Acategory \times \Bcategory   $ and the unique morphism $ \spanS \to 1_{\Acategory, \Bcategory}   $ coincides with (\ref{equation/pairing-functor}).
%
 
%
%
The category $\HoSpanCat$ 
can be thus described as follows.
Its objects are the spans $\spanS=(\Scategory,\sourcefunctor,\targetfunctor)$
whose pairing functor~\eqref{equation/pairing-functor} is an isofibration,
and its morphisms are the simulation 2-cells~\eqref{equation/2-cell-between-spans}
considered up to the following homotopy relation:
%
$\varphi,\psi:\spanS\Rightarrow\spanT$ are declared
\emph{homotopic up to a natural isomorphism},
and we write $\varphi\sim_{\Acategory,\Bcategory}\psi$,
when there exists
a natural isomorphism $\alpha:\varphi\Rightarrow\psi$ between the underlying functors~$\varphi, \psi:\Scategory\to\Tcategory$
of the simulation 2-cells, which is \emph{vertical} in $\Cat/\Acategory\times\Bcategory$, in the expected sense that
\begin{itemize}
\item the natural isomorphism $\sourcefunctor_{\Tcategory}\circ\alpha$ is the identity
on $\sourcefunctor_{\Scategory}$,
\item the natural isomorphism $\targetfunctor_{\Tcategory}\circ\alpha$ is the identity
on $\targetfunctor_{\Scategory}$.
\end{itemize}
The fact that every small category is fibrant in the small category~$\Cat$ ensures
that the composition $\spanS\circ\spanR$ of spans
$\spanR\in\SpanCat(\Acategory,\Bcategory)$ and
$\spanS\in\SpanCat(\Bcategory,\Ccategory)$, given by 
the pullback
\begin{equation}\label{equation/span-composition}
\begin{tikzcd}[column sep=2em, row sep=1ex]
&&
&& 
\Rcategory\times_{\Bcategory}\Scategory
\arrow[lldd,->>,"{\pi_1}"{swap}]
\arrow[rrdd,->>,"{\pi_2}"]
&&
\\
\\
&&
\Rcategory
\arrow[rrdd,->>,"\Rtargetfunctor"{swap}]
\arrow[lldd,->>,"\Rsourcefunctor"{swap}]
\arrow[rrrr,phantom,"{\mathrm{pullback}}"{description}]
&&&&
\Scategory
\arrow[lldd,->>,"\Ssourcefunctor"]
\arrow[rrdd,->>,"\Stargetfunctor"]
\\
\\
\Acategory
&&
&&
\Bcategory
&&
&&
\Ccategory
\end{tikzcd}
\end{equation}
in $\Cat$, is fibrant as soon as $\spanR$ and $\spanS$ are.
The homotopy relations between simulation 2-cells are moreover preserved.
In this way, one obtains a composition functor
\[
\begin{tikzcd}[column sep=2em,row sep=3em]
\circ_{\Acategory,\Bcategory,\Ccategory}
\quad : \quad 
\HoSpanCat({\Bcategory},{\Ccategory})
\times
\HoSpanCat({\Acategory},{\Bcategory})
\arrow[rr]
&&
\HoSpanCat({\Acategory},{\Ccategory})
\end{tikzcd}
\]
equipped with the appropriate associator and unitors required 
to turn the family of homotopy categories~$\HoSpanCat({\Acategory},{\Bcategory})$
into a bicategory~$\HoSpanCat$.
The homotopy construction has the beneficial effect of resolving the obstruction
discussed above, and to define a homotopy variant~$\HoSpanCat$ 
of the bicategory~$\SpanCat$ equipped with a (weak) 2-functor 
\begin{center}
\begin{tikzcd}
\CatEq \arrow[rr] && \HoSpanCat
\end{tikzcd}
\end{center}
from the 2-category~$\CatEq$ obtained by restricting the 2-category~$\CatNat$
to its invertible 2-cells, \emph{i.e.} to natural isomorphisms.
This is sufficient to transport to the bicategory~$\HoSpanCat$
the Seely equivalences living in~$\Cat$.
Following this approach,
Melli{\`e}s obtains a denotational model of \differential linear logic
where every formula~$A$ is interpreted as a small category~$\relsem{A}$
and where every derivation tree 
$\SomeProof{\pi}{A_1 \, , \, A_2 \, , \, \dots \, , \, A_n}$
is interpreted as an isofibration
\begin{equation}\label{equation/isofibration}
\begin{tikzcd}
\relsem{\pi}  \arrow[rr,->>] && \relsem{A_1}\times\cdots\times\relsem{A_n}.
\end{tikzcd}
\end{equation}
%
%
Moreover, the inherently homotopic nature of the construction of $\HoSpanCat$ 
means that the interpretations of proofs $\relsem{\pi}$ provides 
a \emph{homotopy invariant} with respect to cut-elimination.
In other words, two proofs $\pi$ and $\pi'$ equal modulo cut-elimination
have isomorphic interpretations $\relsem{\pi}$ and $\relsem{\pi'}$ in 
the bicategory $\HoSpanCat$,
which means that their interpretations are related by equivalences of categories
(the notion of homotopy in the natural Quillen model structure)
instead of categorical isomorphisms, see~\cite{Mellies19lics} for a discussion.

\paragraph{Canonical fibrant replacements.}

We present the construction of a canonical fibrant replacement of any functor.
Given $ \begin{tikzcd} B \arrow[r,->,"F"] & A \end{tikzcd}$,
we define the \emph{canonical fibrant replacement} of $F$ as follows.
First, we define a category $ \Iso{B} $, whose objects consist of pairs of one object $b$ of $B$ together with an isomorphism $ \alpha : a \to F(b) $ in $A$.
A morphism from $ (b, \alpha) $ to $ (b', \alpha') $ consists of a morphism $ \beta : b \to b' $ and a morphism $ \gamma : a \to a'$ making the following diagram commute
\[\begin{tikzcd}
	a && {a'} \\
	{F(b)} && {F(b')}
	\arrow["\gamma", from=1-1, to=1-3]
	\arrow["\alpha"', from=1-1, to=2-1]
	\arrow["{\alpha'}", from=1-3, to=2-3]
	\arrow["{F(\beta)}", from=2-1, to=2-3]
\end{tikzcd}\,.\]
Note that such a morphism \((\beta,\gamma)\) is in fact uniquely determined by \(\beta\)
because \(\alpha\) and \(\alpha'\) are isomorphisms.
Then, we define \[ \begin{tikzcd}
\Iso{B} \arrow[rr,->>,"\Iso{F}"] && A 
\end{tikzcd} \] by the maps
\begin{align*}
 (b, \alpha : \ty \to F(b)) \qquad \qquad & \mapsto \qquad \qquad \ty 
\\
 (\beta:b\to b', \gamma:a\to a') \qquad \qquad & \mapsto \qquad \qquad \gamma\,.
\end{align*}
The functor \(\Iso{F}\) is indeed an isofibration:
given any isomorphism \(g:a'\to a=\Iso{F}(b,\alpha)\)
we immediately obtain \(\Iso{F}{(\identity{b},g)}=g\).
We moreover obtain an acyclic cofibration \(\morphism*[\iota_F]{B}[acycof]{\Iso{B}}\)
by the inclusion functor sending each object \(b\) to the pair \((b,\identity{F(b)})\):
it is obviously injective of objects and essentially surjective;
it is also fully faithful, because a morphism \((\beta,\gamma)\) in \(\Iso{B}\) is
uniquely determined by \(\beta\).
Conversely, we obtain an acyclic fibration \(\morphism*[\varepsilon_F]{\Iso{B}}[acyfib]{B}\)
by the projection functor sending \((b,\alpha)\) to \(b\) and \((\beta,\gamma)\) to \(\beta\).
We have \(\varepsilon_F\circ\iota_F=\identity{B}\), so that \(\varepsilon_F\) is surjective on objects,
and \(\varepsilon_F\) is fully faithful because every \(\beta:b\to b'\) in \(B\)
determines a unique morphism \((\beta,\alpha'^{-1}\circ F(\beta)\circ\alpha)\)
from \((b,\alpha)\) to \((b',\alpha')\) in \(\Iso{B}\):
hence \(\varepsilon_F\) is an equivalence of categories.
It is moreover an isofibration:
given any isomorphism \(\beta:b'\to b=\varepsilon_F(b,\alpha)\),
the pair \((\beta,\identity{a})\) is an isomorphism
from \((b',F(\beta)^{-1}\circ\alpha)\) to \((b,\alpha)\) in \(\Iso{B}\),
mapped to \(\beta\) by \(\varepsilon_F\).

Note that \(\iota_F\) is a morphism over \(A\),
from \(F\) to \(\Iso{F}\),
\ie\ \(\Iso{F}\circ\iota_F=F\).
On the other hand, \(\varepsilon_F\) is \emph{not}
a morphism from \(\Iso{F}\) to \(F\) in general:
we only have a natural isomorphism
\(\Iso{F}\Rightarrow F\circ\varepsilon_F\),
whose component at \((b,\alpha)\) is \(\alpha\),
naturality being the commutation of the square defining
the morphisms of \(\Iso{B}\).

\paragraph{An explicit shuffle construction of the span homotopy model.}
The original presentation of the span homotopy model of linear logic by
Melli{\`e}s~\cite{Mellies19lics} relied extensively on fibrant objects and
fibrant replacements in $\SpanCat(\Acategory,\Bcategory)$.
Here, we follow alternative recipe to construct an equivalent, but more direct
and explicit interpretation of linear logic.

We declare that a span $\spanS$ is \emph{separately fibrant} 
when the two morphisms $\sourcefunctor{}$ and $\targetfunctor{}$
are isofibrations, that is, fibrations in the natural Quillen model structure
of~$\Cat$.
Note that every span fibrant in~$\SpanCat(\Acategory,\Bcategory)$
is separately fibrant: this is a formal consequence of the fact that every small category is fibrant in~$\Cat$.
On the other hand, the notion of separately fibrant span is strictly more permissive:
the diagonal \(\Delta:\IntGpd\to\IntGpd\times\IntGpd\)
is obviously separately fibrant (each component is the identity on \(\IntGpd\)),
whereas it is not fibrant.

The full subcategory of separately fibrant spans in~$\SpanCat(\Acategory,\Bcategory)$
is noted~$\FibSpanCat(\Acategory,\Bcategory)$.
Note that separately fibrant spans compose horizontally as separately fibrant spans,
and thus define a sub-bicategory $\FibSpanCat$ of the bicategory $\SpanCat$.
%
%
The key idea of our explicit construction is to exploit the following property:
\begin{proposition}
Suppose that $\varphi:\spanR\Rightarrow\spanR'$ and $\psi:\spanS\Rightarrow\spanS'$
are weak equivalences between separately fibrant spans
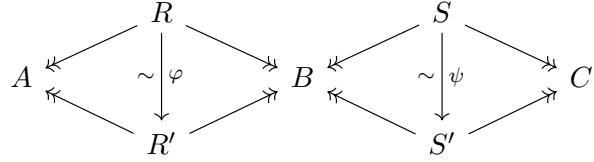
\begin{center}
\begin{tikzcd}[column sep=1.5em,row sep=.4em]
&&
\Rcategory
\arrow[dddd,"{\varphi}","\sim"{swap}]
\arrow[lldd,->>] 
\arrow[rrdd,->>]  
&&&&
\Scategory
\arrow[dddd,"{\psi}","\sim"{swap}] 
\arrow[lldd,->>] 
\arrow[rrdd,->>] 
\\
\\
\Acategory
&&&& 
\Bcategory
&&&& \Ccategory
\\
\\
&& 
\Rcategory'
\arrow[lluu,->>] 
\arrow[rruu,->>] 
&&&&
\Scategory'
\arrow[lluu,->>] 
\arrow[rruu,->>] 
\end{tikzcd}
\end{center}
In that case, the composite 2-cell $\psi\circ\varphi:\spanS\circ\spanR\Rightarrow\spanS'\circ\spanR'$
is a weak equivalence.
\end{proposition}
\begin{proof}
    We first recall that the horizontal composite $ \psi \circ \varphi $ is defined exploiting the universal property of the pullback $ R'\times_B S' $:
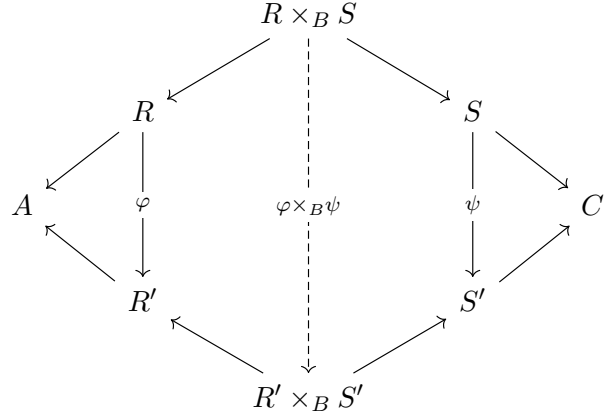
     \[\begin{tikzcd}
	&& {R \times_B S} && \\
	& R && S \\
	A &&&& C \\
	& {R'} && {S'} \\
	 && {R'\times_BS'} && 
	\arrow[from=1-3, to=2-2]
	\arrow[from=1-3, to=2-4]
	\arrow["{\varphi \times_B \psi}"{description}, dashed, from=1-3, to=5-3]
	\arrow[from=2-2, to=3-1]
	\arrow["\varphi"{description}, from=2-2, to=4-2]
	\arrow[from=2-4, to=3-5]
	\arrow["\psi"{description}, from=2-4, to=4-4]
	\arrow[from=4-2, to=3-1]
	\arrow[from=4-4, to=3-5]
	\arrow[from=5-3, to=4-2]
	\arrow[from=5-3, to=4-4]
\end{tikzcd} \] 
Hence, for $ (r,s) \in R \times_B S $, we have that $ \varphi \times_B \psi $ is defined pointwise $ (\varphi \times_B \psi)  \thinspace (r,s) = (\varphi(r), \psi(s))  $. Since the pullback categories are subcategories of the cartesian product, and $ \varphi$ and $ \psi$ are separately weak equivalences, it follows that $ \varphi \times_B \psi  $ is fully-faithful. Since $ \mathcal{R}$ and $ \mathcal{S}$ are separately fibrant and pullbacks preserve isofibrations, we can infer that $ \varphi \times_B \psi$ is  also essentially surjective.
\end{proof}
The result means that the composition of separately fibrant spans
-- that is, the horizontal composition in~$\FibSpanCat$ --
preserves weak equivalences on the left and on the right.
For that reason, it is convenient to describe the bicategory $\HoSpanCat$
as the bicategory $\FibSpanCat$ where every hom-category~$\FibSpanCat(\Acategory,\Bcategory)$
is localized with respect to weak equivalences.
Note that two weak equivalences $\varphi$ and $\psi$ do not necessarily compose horizontally
as a weak equivalence~$\psi\circ\varphi$ in the bicategory~$\SpanCat$
when the spans $\spanR$, $\spanR'$, $\spanS$, $\spanS'$ are not separately fibrant.
%

This observation means that we can perform homotopic constructions 
in the bicategory~$\HoSpanCat$
by taking general spans in $\SpanCat(\Acategory,\Bcategory)$
and by replacing them 
by their \emph{separately fibrant replacements} in~$\FibSpanCat$
instead of their fibrant replacements in~$\HoSpanCat$.
Here, we define a \emph{separately fibrant replacement} of a span~$\spanR$
from $\Acategory$ to $\Bcategory$
as a separately fibrant span~$\spanS$
together with an acyclic cofibration~$\theta:\spanR\Rightarrow\spanS$, as depicted below:
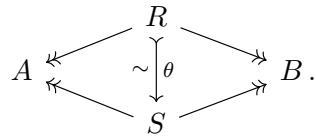
\begin{center}
\begin{tikzcd}[column sep=1.5em,row sep=.5ex]
&&
\Rcategory
\arrow[dddd,>->,"{\theta}","\sim"{swap}] 
\arrow[lldd,->>] 
\arrow[rrdd,->>] 
\\
\\
\Acategory
&&&& 
\Bcategory\,.
\\
\\
&&
\Scategory
\arrow[lluu,->>] 
\arrow[rruu,->>] 
\end{tikzcd}
\end{center}
%
Our main contribution in the present paper is to show 
that the two canonical functors
\begin{center}
\begin{tabular}{c}
\begin{tikzcd}[column sep=1.2em]
\otimes \quad : \quad \Sym\,\Acategory\,\times\,\Sym\,\Acategory
\arrow[rrrr,"{}"]
&&&&
\Sym\,\Acategory 
\end{tikzcd}
\\
\begin{tikzcd}[column sep = 3em]
\distrib_{\Sym{\Bcategory},\Acategory}
\quad : \quad 
\Sym{\,(\Sym\,\Bcategory \times \Acategory)}
\arrow[rr]
&&
\Sym{\,\Bcategory}\times\Sym{\Acategory}
\end{tikzcd}
\end{tabular}
\end{center}
underlying the \emph{contraction rule} and the \emph{promotion rule} of linear logic 
have very simple separately fibrant replacements,
while their fibrant replacements are more complicated.

\begin{figure*}
    \[
        A, B ::= \, 0 \mid \top \mid 1 \mid \bot \mid A\oplus B \mid A \with B 
        \mid A\otimes B \mid A\invamp B \mid {!\,A} \mid {?\, A} \mid X \mid X^{\perp}
    \]
      \begin{gather*}
        \AxiomC{}
        \RightLabel{~axiom}
        \UnaryInfC{$\vdash \, A^{\bot} \, , \, A$}
        \DisplayProof
        \qquad
        \AxiomC{$\vdash \, \Gamma \, , \, A$}
        \AxiomC{$\vdash \, \Delta \, , \, A^{\bot}$}
        \RightLabel{~cut}
        \BinaryInfC{$\vdash \, \Gamma \, , \, \Delta$}
        \DisplayProof
        \qquad
        \AxiomC{$\vdash \, \Gamma \, , \, A \, , \, B \, ,\, \Delta$}
        \RightLabel{~exch}
        \UnaryInfC{$\vdash \, \Gamma \, , \, B \, , \, A \, ,\, \Delta$}
        \DisplayProof
        \\[1ex]
        \AxiomC{$\vdash \, \Gamma \, , \, A$}
        \AxiomC{$\vdash \, \Delta \, , \, B$}
        \RightLabel{~$\otimes$}
        \BinaryInfC{$\vdash \, \Gamma \, , \, \Delta \, , \, A\otimes B$}
        \DisplayProof
        \qquad
        \AxiomC{$\vdash \, \Gamma \, , \, A \, , \, B$}
        \RightLabel{~$\parr$}
        \UnaryInfC{$\vdash \, \Gamma \, , \, A\parr B$}
        \DisplayProof
        \qquad
        \AxiomC{}
        \RightLabel{\quad $\one$}
        \UnaryInfC{$\vdash~\one$}
        \DisplayProof
        \qquad
        \AxiomC{$\vdash \, \Gamma $}
        \RightLabel{~$\bot$}
        \UnaryInfC{$\vdash \, \Gamma \, , \, \bot$}
        \DisplayProof
        \\[1ex]
        \AxiomC{$\vdash \, \Gamma \, , \, A$}
        \AxiomC{$\vdash \, \Gamma \, , \, B$}
        \RightLabel{~$\with$}
        \BinaryInfC{$\vdash \, \Gamma\, , \, A\with B$}
        \DisplayProof
        \qquad
        \AxiomC{$\vdash \, \Gamma \, , \, A_i$}
        \RightLabel{~$\oplus_i$}
        \UnaryInfC{$\vdash \, \Gamma \, , \, A_1\oplus A_2$}
        \DisplayProof
        \qquad
        \AxiomC{}
        \RightLabel{~$\top$}
        \UnaryInfC{$\vdash\,\Gamma\,,\,\top$}
        \DisplayProof
        \\[1ex]
        \AxiomC{$\vdash \, \wn \Gamma \, , \, A$}
        \RightLabel{~$\oc$}
        \UnaryInfC{$\vdash \, \wn \Gamma \, , \, \oc A$}
        \DisplayProof
        \qquad
        \AxiomC{$\vdash \, \Gamma \, , \, A$}
        \RightLabel{~$\wn$}
        \UnaryInfC{$\vdash \, \Gamma \, , \, \wn A$}
        \DisplayProof
        \qquad
        \AxiomC{$\vdash \, \Gamma \, , \, \wn A \, , \, \wn A$}
        \RightLabel{~contr}
        \UnaryInfC{$\vdash \, \Gamma \, , \, \wn A$}
        \DisplayProof
        \qquad
        \AxiomC{$\vdash \, \Gamma $}
        \RightLabel{~weak}
        \UnaryInfC{$\vdash \, \Gamma \, , \, \wn A$}
        \DisplayProof
    \end{gather*}
    \caption{Formulae and rules of Linear Logic.}  \hrulefill
    \label{figure/LL}  
\end{figure*}

We find convenient to introduce the following notations.
Given small category \(\Xcategory\) and a list \(\Gamma=\Acategory_1,\dotsc,\Acategory_n\) of small categories,
we denote by \(\morphism*[\tysf]{\Xcategory}[spanmaps]{\Gamma}\) any 
\(n\)-fold span:
that is given by a family of functors
\(\morphism*[f_i]{\Xcategory}{\Acategory_i}\) for \(1\le i\le n\)
-- this induces a functor to the product \(\prod\Gamma=\Acategory_1\times\cdots\times\Acategory_n\)
by tupling \(\morphism*[\tuple{f_1,\dotsc,f_n}]{\Xcategory}{\prod\Gamma}\).
We generalize the notion of separately fibrant spans to such
\(n\)-fold spans: we say \(\tysf\) is separately fibrant,
and we write \(\morphism*[\tysf]{\Xcategory}[sepfib]{\Gamma}\),
when each \(f_i\) is a fibration.

Now, instead of a fibration
\(\morphism*{\relsem{\pi}}[fibrant]{\relsem{A_1}\times\cdots\times\relsem{A_n}}\),
we will interpret a proof \(\SomeProof{\pi}{\Gamma=A_1,\dotsc,A_n}\) as a separately fibrant \(n\)-fold span:
\begin{equation}
    \label{equation/sepfibsem}
    \morphism{\relsem{\pi}}[sepfib]{\relsem{\Gamma}=\relsem{A_1},\dotsc,\relsem{A_n}}\,.
\end{equation}
Note that one can always recover a fibration from a proof:
setting \(A=A_1\parr\cdots\parr A_n\) and
\(\SomeProof{\pi^\parr}{A}\) the proof obtained from \(\pi\)
by iterating the \((\parr)\) rule,
the separately fibrant span \(\relsem{\pi^\parr}\) consists of a single fibration:
\begin{equation}
    \label{equation/fibrantsem}
    \morphism{\relsem{\pi^\parr}}[fibrant]{\relsem{A}}\,.
\end{equation}
In a sense, distinguishing between \eqref{equation/sepfibsem}
and \eqref{equation/fibrantsem} is precisely what allows us
to provide a very simple and concrete interpretation of proofs.

\section{Interpretation of linear logic}\label{section/LL}
We explain how to interpret the formulae and proofs of classical linear logic 
which we describe in the usual syntactic way using a (one sided) sequent calculus,
that we recall in \cref{figure/LL}.
We suppose given an interpretation of each atomic formula~$X$ of linear logic
as a small category~$\sem{X}\in\Cat$.
The interpretation is then performed in two stages.
First, the formulae of classical linear logic
are interpreted by structural induction:
\[
\begin{array}{ccccc}
\sem{0} = \sem{\top} = \mathbf{0}
& \quad &
\sem{ A\oplus B} = \sem{ A\with B} = \sem{A}+\sem{B}
&\quad&
\sem{X^{\perp}}=\sem{X}
\\
\sem{1} = \sem{\bot} = \mathbf{1}
& \quad &
\sem{ A\otimes B} = \sem{ A\invamp B} = \sem{A}\times\sem{B}
&\quad&
\sem{{!\,A}}=\sem{{?\, A}}=\Sym\sem{A}
\end{array}
\]
so that each sequent \(\Gamma=A_1,\dotsc,A_n\) induces a tuple of small categories
which we simply write \(\sem{\Gamma}=\sem{A_1},\dotsc,\sem{A_n}\).
The interpretation of the derivation tree~$\pi$ is defined by structural induction,
by inspecting each rule of the sequent calculus.
    Note that the interpretation of the exchange rule (exch) is straightforward:  
    we let the same permutation act on the tuple of isofibrations.
From now on, we find more convenient and easier for the reader
to drop the semantic bracket~$\sem{-}$ from our notations.

%
%
\paragraph{Axiom.}
The axiom rule is simply interpreted as the identity span:
\begin{equation}\label{equation/id-span}
\begin{tikzcd}[column sep=5em]
    \Acategory
    & \Acategory \arrow[l,fibrant,"\id{}"{swap}] \arrow[r,fibrant,"\id{}"]
    & \Acategory
\end{tikzcd}
\end{equation}
It is important to note that although this span is separately fibrant
because each of its component is an isofibration,
it is in general not fibrant in the model category~$\HoSpanCat{(\Acategory,\Acategory)}$
because the diagonal functor $\Acategory\to \Acategory\times \Acategory$ itself is not an isofibration in general.
This interpretation thus substantially improves the original model~\cite{Mellies19lics}
where the axiom rule was interpreted as the fibrant replacement of the diagonal
defined by the category of isomorphisms of \(\Acategory\), equipped with is source and target functors.

\paragraph{Cut.}
Consider two proofs
$\SomeProof{\pi_1}{\Gamma \, , \, A}$
and
$\SomeProof{\pi_2}{\Delta \, , \, A^{\bot}}$
interpreted as separately fibrant spans
\begin{center}
\begin{tikzcd}[column sep = 1.8em]
\Gamma
&&
{{\pi_1}} \arrow[rr,fibrant] \arrow[ll,sepfib]
&&
{A}
\end{tikzcd}
\quad\quad\quad\quad
\begin{tikzcd}[column sep = 1.8em]
{A}
&&
{{\pi_2}} \arrow[ll,fibrant] \arrow[rr,sepfib]
&&
{\Delta}
\end{tikzcd}
\end{center}
The interpretation of the derivation tree $\pi$ obtained by cutting $\pi_1$ and $\pi_2$,
is obtained by computing the pullback \( \pi := \pi_1\times_A\pi_2 \) of the cospan
\begin{tikzcd}[cramped] \pi_1 \arrow[r,fibrant] & A & \pi_2 \arrow[l,fibrant] \end{tikzcd}
consisting of the two isofibrations on the cut formula.
Isofibrations are closed under pullbacks, and this thus yields a separately fibrant span
consisting of the two projection maps:
\begin{center}
\begin{tikzcd} \pi_1 && \pi_1\times_A\pi_2  \arrow[ll,fibrant,"{\fstproj}"{swap}]\arrow[rr,fibrant,"{\sndproj}"] && \pi_2 \end{tikzcd}
\end{center}
%
The desired interpretation is then obtained
\[\begin{tikzcd} 
\Gamma && \arrow[ll,sepfib] \pi_1 && \pi_1\times_A\pi_2  \arrow[ll,fibrant,"{\fstproj}"{swap}]\arrow[rr,fibrant,"{\sndproj}"] && \pi_2\arrow[rr,sepfib] && \Delta
\end{tikzcd}
\]
by postcomposing the projection maps
with the components 
of the separately fibrant spans \(\sem{\pi_1}\) and \(\sem{\pi_2}\).

\paragraph{Additives}
The interpretation of additives is easy and does not require any explicit fibrant replacement.
Consider two proofs
$\SomeProof{\pi_1}{\Gamma \, , \, A}$
and
$\SomeProof{\pi_2}{\Gamma \, , \, B}$
interpreted as separately fibrant spans
\begin{center}
\begin{tikzcd}[column sep = 1.8em]
\Gamma
&&
{{\pi_1}} \arrow[rr,->>] \arrow[ll,->>,"\circ" marking]
&&
{A}
\end{tikzcd}
\quad\quad\quad\quad
\begin{tikzcd}[column sep = 1.8em]
\Gamma
&&
{{\pi_2}} \arrow[rr,->>] \arrow[ll,->>,"\circ" marking]
&&
{B}
\end{tikzcd}
\end{center}
%
%
The interpretation of the derivation tree obtained by combining $\pi_1$ and $\pi_2$
using a with rule $(\with)$ is defined as the sum of the two interpretations
\begin{center}
\begin{tikzcd}
\Gamma
&&
{{\pi_1}+{\pi_2}} \arrow[rr,->>] \arrow[ll,->>,"\circ" marking]
&&
{A + B}
\end{tikzcd}
\end{center}
where the context maps from~$\pi_1+\pi_2$ to $\Gamma$ 
are defined by the universality property of the sum,
while the map from~$\pi_1+\pi_2$ to $A+B$ is defined
as the image of the two maps from $\pi_1$ to $A$
and from $\pi_2$ to $B$ by the sum functor.
The universal construction as well as the sum functor
preserve isofibrations, and the resulting span
is thus separately fibrant, as indicated in the diagram.

Similarly, the proof reduced to a $(\top)$ rule is interpreted 
as the separately fibrant span
\begin{center}
\begin{tikzcd}
\Gamma
&&
{\zerocategory} \arrow[rr,->>] \arrow[ll,->>,"\circ" marking]
&&
{\zerocategory}
\end{tikzcd}
\end{center}
where the context isofibrations are defined by the universality property
of the initial category~$\zerocategory$
and where the isofibration from~$\zerocategory$ to~$\zerocategory$
is the identity.

On the other hand, the plus rules $(\oplus_1)$ and $(\oplus_2)$ 
are interpreted by leaving the context maps unchanged
and postcomposing the map to the interpretation
of the principal formula~$A$ or~$B$
with the left and right injections, which are isofibrations:
\begin{center}
\begin{tikzcd}
{A}\arrow[rr,->>]
&&
{A}+{B}
&\text{and}&
{B}\arrow[rr,->>]
&&
{A}+{B}
\,.
\end{tikzcd}
\end{center}

\paragraph{Multiplicatives (tensor, one).}
The interpretation is easy and does not require any explicit fibrant replacement.
Consider two proofs
$\SomeProof{\pi_1}{\Gamma \, , \, A}$
and
$\SomeProof{\pi_2}{\Delta \, , \, B}$
interpreted as separately fibrant spans
\begin{center}
\begin{tikzcd}[column sep = 1.8em]
\Gamma
&&
{{\pi_1}} \arrow[rr,fibrant] \arrow[ll,sepfib]
&&
{A}
\end{tikzcd}
\quad\quad\quad\quad
\begin{tikzcd}[column sep = 1.8em]
{\Delta}
&&
{{\pi_2}} \arrow[ll,sepfib] \arrow[rr,fibrant]
&&
{B}
\end{tikzcd}
\end{center}
The interpretation of the derivation tree $\pi$ obtained by combining $\pi_1$ and $\pi_2$
using a tensor rule $(\otimes)$ is defined as the cartesian product~$\pi:=\pi_1\times\pi_2$
of the two interpretations with the separately fibrant span
\begin{center}
\begin{tikzcd}[row sep = 0em]
{\Gamma} && {\pi_1}\arrow[ll,sepfib]
\\
&& && {{\pi_1}\times{\pi_2}} \arrow[llu,fibrant,"{\fstproj}"{swap}]\arrow[lld,fibrant,"{\sndproj}"]\arrow[rr,fibrant] && {A}\times{B}
\\
{\Delta} && {\pi_2}\arrow[ll,sepfib]
\end{tikzcd}
\end{center}
obtained by postcomposing the projection maps with the components
of the separately fibrant spans~\(\sem{\pi_1}\) and \(\sem{\pi_2}\).
Here, we use a general property of model categories
that the class $\mathcal{F}$ of fibrations (= the  isofibrations
in the case of the natural Quillen model structure on $\Cat$)
is closed under cartesian product.
We also use the fact that every object is fibrant,
and that the projection maps $\fstproj$ and $\sndproj$ are thus isofibrations.
%
The proof reduced to the rule $(1)$ introducing
the tensor unit is interpreted by the identity map~$\onecategory\to\onecategory$ 
which defines an isofibration on the category~$\onecategory$.

\paragraph{Multiplicatives (par, bot).}
One interesting aspect of our presentation of the homotopy span model
is that the interpretation of the par rule ($\parr$) is nontrivial and involves a fibrant replacement.
Consider a proof
$\SomeProof{\pi}{\Gamma \, , \, A \, , \, B}$
interpreted as a separately fibrant span
\begin{center}
\begin{tikzcd}[row sep = -.6em, column sep = 1.8em]
&&&& A
\\
\Gamma
&&
{{\pi}} \arrow[rru,->>]\arrow[rrd,->>] \arrow[ll,->>,sepfib]
\\
&&&&
B
\end{tikzcd}
\end{center}
The functor $\pi\to A\times B$ obtained by the universality property of the cartesian product in~$\Cat$
is not necessarily an isofibration, but it can be factored in the Quillen model category
as an acyclic cofibration followed by an isofibration:
\begin{center}
\begin{tikzcd}[row sep = 0em, column sep = 1.8em]
{{\pi}} \arrow[rr,>->,"\iota_\pi","\sim"'] 
&&
{{\Iso{\pi}}} \arrow[rr,->>]
&&
A\times B
\end{tikzcd}
\end{center}
where $ \Iso{\pi} \to A\times B$ is defined as the canonical fibrant replacement of $\pi\to A\times B$.
Note that the category~$\Iso{\pi}$ comes with the acyclic fibration~$\varepsilon_\pi$
of the canonical fibrant replacement:
\begin{center}
\begin{tikzcd}[row sep = 0em, column sep = 1.8em]
{{\pi}}
&&
{{\Iso{\pi}}} \arrow[rr,->>] \arrow[ll,->>,"\varepsilon_\pi"{swap},"\sim"]
&&
A\times B
\end{tikzcd}
\end{center}
The interpretation of the derivation tree $\SomeProof{\pi'}{\Gamma \, , \, A \parr B}$
obtained by applying the par ($\parr$) rule is then obtained by taking 
the category $\Iso{\pi}$ as support and postcomposing
\begin{center}
\begin{tikzcd}[row sep = 0em, column sep = 1.8em]
\Gamma
&&
{{\pi}} \arrow[ll,sepfib] 
&&
{{\Iso{\pi}}} \arrow[rr,->>] \arrow[ll,->>,"\varepsilon_\pi"{swap},"\sim"]
&&
A\times B
\end{tikzcd}
\end{center}
with the context maps of the original separately fibrant span.

On the other hand, the interpretation of the par unit $(\bot)$ is easy 
to define and does not require any explicit fibrant replacement.
Given a proof~$\SomeProof{\pi}{\Gamma}$ interpreted as the separately fibrant span
\begin{center}
\begin{tikzcd}[row sep = 0em, column sep = 1.8em]
\Gamma
&&
{{\pi}} \arrow[ll,sepfib] 
\end{tikzcd}
\end{center}
the interpretation of the proof~$\SomeProof{\pi'}{\Gamma\,,\,\bot}$ 
obtained after the rule~$(\bot)$ is simply defined by adding a functor 
to the terminal category~$\onecategory$ in the following way:
\begin{center}
\begin{tikzcd}[row sep = 0em, column sep = 1.8em]
\Gamma
&&
{{\pi}} \arrow[ll,sepfib] \arrow[rr,->>]
&&
\onecategory
\end{tikzcd}
\end{center}
Here, one uses the specific property of the natural Quillen model on~$\Cat$
that every object~$\pi$ is fibrant, which precisely means that the unique functor $\pi\to\onecategory$ 
is an isofibration.
%
%


\paragraph{Exponentials (dereliction, weakening).}
The promotion rules will be treated 
in dedicated \cref{section/promotion-rule},
but it is already easy to describe the rest of the rules.
Consider one proof $\pi\vdash \Gamma, A$ interpreted as a separately fibrant span
\begin{center}
\begin{tikzcd}[column sep = 1.8em]
\Gamma
&&
{{\pi}} \arrow[rr,fibrant] \arrow[ll,sepfib]
&&
{A}
\end{tikzcd}
\end{center}
the proof $\vdash \Gamma, ?A$ 
obtained after a dereliction rule ($\wn$) is interpreted 
by keeping the same underlying category and postcomposing the span
\begin{center}
\begin{tikzcd}[column sep = 1.8em]
\Gamma
&&
{{\pi}} \arrow[rr,fibrant] \arrow[ll,sepfib]
&&
{A} \arrow[rr,fibrant,"{\eta_{\Acategory}}"] 
&&
\Sym{\,\Acategory}
\end{tikzcd}
\end{center}
with the isofibration $\eta_{\Acategory}:\Acategory\to\Sym{\,\Acategory}$
which sends any object $a$ of ${A}$ to the singleton word $\whynotword{a}$.
Similarly, given a proof $\pi\vdash \Gamma$ interpreted as a separately fibrant span, 
the proof  $\pi'\vdash \Gamma,?A$ obtained after the weakening rule (weak) 
is interpreted by projecting~$\pi$ on the terminal category~${\onecategory}$ 
with the isofibration $\pi\to\onecategory$ and then composing
\begin{center}
\begin{tikzcd}[column sep = 1.8em]
\Gamma
&&
{\pi} \arrow[rr,fibrant] \arrow[ll,sepfib]
&&
\onecategory
\arrow[rr,fibrant]
&&
\Sym{\,\Acategory}
\end{tikzcd}
\end{center}
with the isofibration $\onecategory\to\Sym{\,\Acategory}$
which sends the unique object $\ast$ of ${\onecategory}$
to the empty word $\whynotword{}$ in~$\Sym{\,\Acategory}$.

\paragraph{Exponentials (contraction).}
For the contraction, we follow a treatment similar to that of the \(\parr\).

Given a proof $ \pi \vdash \Gamma \, , \, \wn A \, , \, \wn A $ and its interpretation as a separately fibrant span 
\[\begin{tikzcd} [row sep=0pt]
    &&&&\SymWhyNot{A}
    \\
    \Gamma 
    &&
    \pi 
    \arrow[ll,sepfib,"\tysf"']
    \arrow[rru,fibrant,"f"]
    \arrow[rrd,fibrant,"g"']
    \\
    &&&&\SymWhyNot{A}
\end{tikzcd} \]
we need to define an isofibration
\begin{tikzcd}[cramped]
\Gamma
&&
{{\contra{f}{g}}}
\arrow[rr,fibrant]
\arrow[ll,sepfib]
&&
{ \SymWhyNot{A}}
\end{tikzcd}. 
We first use the concatenation 
$\morphism*{\SymWhyNot\Acategory\times\SymWhyNot\Acategory}{\SymWhyNot\Acategory}$
to obtain:
\[
    \begin{tikzcd}
        \pi \arrow[rr,"\tuple{f,g}"]
        && \SymWhyNot\Acategory\times\SymWhyNot\Acategory \arrow[rr,"\concat"]
        && \SymWhyNot\Acategory
    \end{tikzcd}
\]
which, however, is not an isofibration.
Again, we consider its canonical fibrant replacement:
\begin{center}
\begin{tikzcd}[row sep = 0em, column sep = 1.8em]
{{\pi}} \arrow[rr,acycof',"\iota"] 
&&
{{\contra{f}{g}}} \arrow[rr,fibrant]
&&
\SymWhyNot\Acategory
\end{tikzcd}
\end{center}
which moreover comes with the acyclic fibration
\(\morphism*[\varepsilon]{\contra{f}{g}}[acyfib]{\pi}\).
The interpretation of the results of applying the contraction rule to \(\pi\)
is then obtained by taking the category $\contra{f}{g}$ as support and postcomposing
\begin{center}
\begin{tikzcd}[row sep = 0em, column sep = 1.8em]
\Gamma
&&
{{\pi}} \arrow[ll,sepfib] 
&&
{{\contra{f}{g}}} \arrow[rr,fibrant] \arrow[ll,acyfib',"\varepsilon"{swap}]
&&
\SymWhyNot\Acategory
\end{tikzcd}
\end{center}
with the context maps of the original separately fibrant span.

\section{The shuffle factorization theorem}
\label{section/factorization}

Before carrying on, we find convenient to introduce two different notations
for the symmetric strict monoidal category~$\Sym\,\Acategory$
freely generated by a small category~$\Acategory$:
\begin{center}
$\SymPlus{\,A}$ with objects finite words of objects of~$\Acategory$ noted $\genbangword{a}{n}$;
\end{center}
\begin{center}
$\SymProd{\,A}$ with objects finite words  of objects of~$\Acategory$ noted $\genwhynotword{a}{n}$.
\end{center}
Morphisms in $\SymProd{A}$ are the same as before in $\Sym{A}$, but a morphism
in $\SymPlus{(\Acategory)}$ 
\begin{center}
\begin{tikzcd}
\genbangmorph{\sigma}{f}{n}
\quad : \quad \genbangword{a}{n}\arrow[rr] && \genbangword{a'}{n}
\end{tikzcd}
\end{center}
is a tuple consisting of a permutation $\sigma\in \permutationgroup{n}$
together with a family of morphisms $f_i:a_i\to a'_{\sigma(i)}$ in $\Acategory$.
%
%
%
Our reason for having distinct conventions is that we see $\SymPlus{(\Acategory)}$ as a subcategory
of the free cocartesian category $\FamAface{(\Acategory)}$
generated by $\Acategory$ -- with the same objects as \(\Sym{(\Acategory)}\),
but with morphisms
\begin{center}
\begin{tikzcd}
\genbangmorph{\sigma}{f}{n}
\quad : \quad \genbangword{a}{n}\arrow[rr] && \genbangword{a'}{p}
\end{tikzcd}
\end{center}
consisting of a function $\sigma: \{1,\dotsc,n\}\to \{1,\dotsc,p\}$
together with a family of morphisms $f_i:a_i\to a'_{\sigma(i)}$ in $\Acategory$
for $1\le i\le n$ --
and $\SymProd{(\Acategory)}$
as a subcategory of the free cartesian category $\FamBface{(\Acategory)}$
generated by $\Acategory$ --
defined symmetrically.
We then choose to redefine the interpretation of exponential formulae,
setting $\sem{\wn A}=\SymWhyNot{\sem A}$ and $\sem{\oc A}=\SymBang{\sem A}$.

%
Recall that the category $\FamBface{(\Acategory)}$ defines a Grothendieck opfibration
\begin{center}
\begin{tikzcd}[column sep=1.5em]
q \quad : \quad \FamBface{\Acategory}\arrow[rr] && \Set
\end{tikzcd}
\end{center}
Given a function $f:J\to I$ and a family $\familyBface{a_i}{i}{I}$
indexed by $I$, we write
\[
\push{f} :  \familyBface{a_{i\in I}}{i}{I}  \to \familyBface{a_{f(j)}}{j}{J} 
\]
for the canonical map.
We are now ready to define the notions of \emph{shuffle} and \emph{deshuffle
isomorphisms}
in the symmetric monoidal categories~$\SymBang\Acategory$
and~$\SymWhyNot\Acategory$.
As a symmetric monoidal category, $\SymBang\Acategory$ can be seen 
as a symmetric operad with the same objects, also noted $\SymBang\Acategory$,
whose $n$-ary morphisms $x_1,\dots x_n\to x$
are defined as the morphisms $x_1\otimes\cdots\otimes x_n\to x$
of the category.
%
Given a sequence of $n$ natural numbers $k_1, \dotsc, k_n$
with sum the natural number~$k=k_1+\cdots+k_n$ 
seen as totally ordered sets,
a $(k_1,\dotsc,k_n)$-shuffle~$\sigma$ is a bijective function
\begin{center}
\begin{tikzcd}[column sep = 1.5em]
k_1+\dots+k_n \arrow[rr,"\sigma"] &&  k
\end{tikzcd}
\end{center}
such that each restriction
\begin{center}
\begin{tikzcd}[column sep = 1.5em]
k_i\arrow[rr,"{\mathbf{inj}_i}"] && k_1+\dots+k_n \arrow[rr,"\sigma"] &&  k
\end{tikzcd}
\end{center}
is a monotone function.
A \emph{shuffle isomorphism} 
in the operad $\SymBang\Acategory$
is defined as a morphism
\begin{center}
\begin{tikzcd}
    \bangmorph{\sigma}{\identity{},\dotsc,\identity{}}
\quad : \quad 
\genbangword{a^1}{k_1} \,,\, \cdots \,,\, \genbangword{a^n}{k_n}
\arrow[rr]
&&
\genbangword{a'}{k_1+\dots+k_n}
\end{tikzcd}
\end{center}
where the permutation $\sigma\in\permutationgroup{k_1+\dots+k_n}$
is a $(k_1,\dotsc,k_n)$-shuffle.

\begin{proposition}
Every $n$-ary morphism $f:x_1,\cdots,x_n\to y$ 
in the operad $\SymBang\Acategory$
factors uniquely as $f=\sigma\circ (h_1,\cdots,h_n)$
where $\sigma:y_1,\cdots,y_n\to y$ is a shuffle isomorphism 
and each $h_k:x_k\to y_k$ is a unary morphism, for $1\leq k\leq n$.
\end{proposition}
Symmetrically, $\SymWhyNot\Acategory$ defines a symmetric cooperad
where a \emph{deshuffle isomorphism}
\begin{center}
\begin{tikzcd}
\whynotmorph{\sigma}{\identity{},\dotsc,\identity{}}
\quad : \quad 
\genwhynotword{a'}{k_1+\dots+k_n}
\arrow[rr]
&&
\genwhynotword{a^1}{k_1} \,,\, \cdots \,,\, \genwhynotword{a^n}{k_n}
\end{tikzcd}
\end{center}
is defined as a $n$-ary morphism
whose counterpart (or reversal)
in the opposite symmetric operad $(\SymWhyNot\Acategory)^{\,op}=\SymBang(\Acategory^{\,op})$
is a shuffle isomorphism.

Given a sequence of categories $ \Gamma = A_1,\dots, A_n$, we observe that the product category $ \SymWhyNot{\Gamma}^\times =  \SymWhyNot{A_1} \times \cdots \times \SymWhyNot{A_k}  $ is a symmetric strict monoidal category, with tensor product defined by pointwise list concatenation. The shuffle factorization system lift to $ \SymWhyNot{\Gamma}^\times $ componentwise. A deshuffle isomorphism in $ \SymWhyNot{\Gamma}^\times $ consists of a tuple $  (\sigma_1, \dots, \sigma_n) $ where $ \sigma_i$ is a deshuffle isomorphism in $ \SymWhyNot{A_i}$.

\section{The promotion rule}\label{section/promotion-rule}
We now consider the \emph{promotion} rule, whose interpretation
exploits the shuffle factorization observed in~\S\ref{section/factorization}.
%
Given a context $ \Gamma = A_1 ,\dots, A_n$ we use the notations $ \wn \Gamma =  \wn A_1 ,\dots,  \wn A_n$
and $ \SymWhyNot{\Gamma} = \SymWhyNot{\Acategory_1} , \dots,  \SymWhyNot{\Acategory_n}$.
Given a proof $\SomeProof{\pi}{\wn \Gamma, B}$,
we consider its interpretation as a separately fibrant span 
\begin{center}
\begin{tikzcd}
 \SymWhyNot{\Gamma} 
&&
{{\pi}}
\arrow[rr,fibrant,"g"]
\arrow[ll,sepfib, "\tysf"']
&&
{B}
\end{tikzcd}
\end{center}
and describe how to build the separately fibrant span
interpreting the result of the promotion of \(\pi\):
\begin{center}
\begin{tikzcd}
 \SymWhyNot{\Gamma} 
&&
{{\pi^\oc}} \arrow[rr,fibrant] \arrow[ll,sepfib]
&&
{\SymBang{B}}
\end{tikzcd}.
\end{center}
First, we observe that the functorial action of $  \SymBang{}$ preserves isofibrations:
given an isofibration $\morphism*[f]{\Acategory}[fibrant]{\Bcategory}$
the functor
$\morphism*[\SymBang{f}]{\SymBang{\Acategory}}{\SymBang{\Bcategory}}$
defined by the map 
\( \bangword{\ty_1,\dots, \ty_k} \mapsto \bangword{f(\ty_1),\dots, f(\ty_k)} \)
is also an isofibration.
Hence, we have an isofibration
$$\SymBang{g} \quad : \quad \begin{tikzcd}[cramped] {{\SymBang{\pi}}}  \arrow[r,fibrant] & \SymBang{\Bcategory} \end{tikzcd}$$
However, we also need to deal with the context span
$\tysf:\begin{tikzcd}[cramped] {{\pi}}  \arrow[r,sepfib] & \SymWhyNot{\Gamma} \end{tikzcd}$.

As a first approach, we could try to exploit the universal property of $ \SymBang{\pi}$.
Given a symmetric (unbiased) monoidal category $X $ and a functor 
$f:\begin{tikzcd}[cramped] {{\Acategory}}  \arrow[r] & X \end{tikzcd}$,
by the fact that $\SymBang{\pi} $ is the free symmetric strict monoidal category on $ \pi$,
we obtain a functor
$\monlift{f}:\begin{tikzcd}[cramped] {\SymWhyNot{\Acategory}}  \arrow[r]& X \end{tikzcd} $,
induced by the map
\( \seqdots{\ty}{1}{k} \mapsto f(a_1) \otimes \cdots \otimes f(a_k) \),
making the following diagram commute
\[\begin{tikzcd}
	A
	\arrow[rr, "{\eta_A}"]
	\arrow[d, "f"']
        &&
        {\SymWhyNot{A}}
        \arrow[dll,"{\monlift{f}}"]
        \\
	X
\end{tikzcd}.\]
Hence, we can apply this construction to our span
$\tysf:\begin{tikzcd}[cramped] {{\pi}}  \arrow[r,sepfib] & \SymWhyNot{\Gamma} \end{tikzcd}$,
componentwise, and obtain a span
\begin{center}
\begin{tikzcd}[column sep = 1.5em] 
\monlift{\tysf} \quad : \quad {\SymBang{\pi}} \arrow[rr,spanmaps] && \SymWhyNot{\Gamma}
\end{tikzcd}
\end{center}
However, the resulting span has no reason to be separately fibrant.
A simple example is given by applying the construction to the identity
$ \identity{\SymBang{A}}$
since $\monlift{\identity{\SymBang{A}}}:\begin{tikzcd}[cramped] {\SymBang{\SymBang{\Acategory}}}  \arrow[r] & \SymBang{\Acategory} \end{tikzcd}$
coincides with the monad multiplication
\( \seqdots{\tyl}{1}{k}  \mapsto  \tyl_1  \otimes \cdots \ \otimes \tyl_k \),
which is not an isofibration in general.

In order to fix this, we shall perform a replacement of the span
$\monlift{\tysf}$ by an appropriate construction,
that we shall call the \emph{shuffle category} on $ \tysf $,
denoted as $  \Shuffle_{ \tysf  }$.
As we will see, this construction comes equipped with an acyclic cofibration (an injective on objects equivalence of categories)
\begin{tikzcd}[cramped] \SymBang{\pi} \arrow[r,>->,"{\sim}"] & {\Shuffle_{\tysf}} \end{tikzcd}
and a separately fibrant span
$\sfiblift{\tysf} :\begin{tikzcd}[cramped] {\Shuffle_{\tysf}} \arrow[r,sepfib] & \SymWhyNot{\Gamma} \end{tikzcd}$
making the following diagram commute 
\begin{equation}\label{equation/commutative-lifting-diagram}
\begin{tikzcd}
        \SymBang{\pi}
        \arrow[rr,acycof]
        \arrow[dr,spanmaps,"\monlift{\tysf}"{swap}]
        &&
	\Shuffle_{\tysf}
        \arrow[dl,sepfib,"\,\sfiblift{\tysf}"]
        \\
        &
	\SymWhyNot{\Gamma}
\end{tikzcd}
\end{equation}
as well as with an acyclic isofibration
$\morphism*{\Shuffle_{\tysf}}[fibrant,acyclic]{\SymBang{\pi}}$
whose postcomposition with 
\(\SymBang{g}:\begin{tikzcd}[cramped]\SymBang{\pi}\arrow[r,fibrant]&\SymBang{B}\end{tikzcd}\)
will give the remaining component of the separately fibrant span \(\pi^\oc\).

\paragraph{The shuffle category}
We associate to every span 
\(\morphism*[\tysf]{\pi}[spanmaps]{\SymWhyNot{\Gamma}}\),
the \emph{shuffle category} \(\Shuffle_{\tysf}\) defined as follows:
\begin{itemize}
\item an object of \(\Shuffle_{\tysf}\) is a pair consisting of an object
\[
 \bangword{p_1,\dots, p_k}
 \in \SymBang{ \pi  }
\]
together with a componentwise deshuffle isomorphism
\begin{center}
\begin{tikzcd}[column sep = 1em]
\varphi = (\varphi_1,\dots,\varphi_n) \, \quad : \quad \,
 \gamma \arrow[rr] &&   \tysf(p_1) \otimes \dotsc \otimes \tysf (p_k)
\end{tikzcd}
\end{center}
from an object 
\begin{center}
$\gamma \,\in\,\SymWhyNot{\Gamma}^\times = \SymWhyNot{A_1} \times \cdots \times \SymWhyNot{A_n}$
\end{center}
where each $\varphi_i$ is a deshuffle isomorphism in the category $\SymWhyNot{A_i}$,  for $1\leq i\leq n$.
%
%
%
\item a morphism from an object
\begin{center}
$\bangword{ p_1,\dots, p_k} \,\in\, \SymBang{\pi}$
with componentwise deshuffle isomorphism
$\varphi \, : \, \gamma \to \tysf(p_1) \ \otimes \ \dotsc  \ \otimes \ \tysf (p_k)$
\end{center}
to an object
\begin{center}
$\bangword{ p'_1,\dots, p'_k} \,\in\, \SymBang{\pi}$
with componentwise deshuffle isomorphism
$\varphi' \, : \, \gamma' \to \tysf(p'_1) \ \otimes \ \dotsc  \ \otimes \ \tysf (p'_k)$
\end{center}
is defined as a triple consisting of a permutation $\sigma\in\permutationgroup{k}$
together with 
\begin{center}
a family of morphisms
${{h_i} \, : \, {p_i} \to {p'_{\sigma i}}}$ in the category~$\pi$
for $1\le i\le k$,
\end{center}
and a morphism
$\bar{h} \, : \, {\gamma} \to {\gamma'}$
in $\SymWhyNot{\Gamma}^\times$ making the diagram commute:
\begin{equation}
    \label{equation/Shuffle}
    \begin{tikzcd}
    \Xmonobjc
    \arrow[d,"\bar{h}"{swap}]\arrow[rrrr,"{\varphi}"]
    &&
    &&
    {\tysf(p_1) \Xtensor \cdots \Xtensor \tysf(p_k) }
    \arrow[d,"{\tysf(h_1)\Xtensor\cdots \Xtensor \tysf(h_k)}"]
    \\
    \Xmonobjcp\arrow[rr,"{\varphi'}"]
    &&
    {\tysf (p'_1) \Xtensor\cdots \Xtensor \tysf (p'_k)}
    \arrow[rr,"{\push{\sigma}}"]
    &&
    {\tysf (p'_{\sigma(1)}) \Xtensor\cdots \Xtensor \tysf (p'_{\sigma(k)})}
    \end{tikzcd}
\end{equation}
\end{itemize}
Note that the data of the morphism $\bar{h}$ is redundant since 
it always exists and is uniquely determined in the category $\SymWhyNot{\Gamma}^\times$
by the sequence of morphisms $h_1,\dotsc,h_k$ because $\varphi$ and $\varphi'$ and $\push{\sigma}$ are isomorphisms.
We thus generally denote such a morphism as a pair
\(\genbangmorph{\sigma}{h}{k}\).

\paragraph{The acyclic cofibration}

We now define the equivalence of categories 
\begin{equation}\label{equation/shufflecofibration}
\begin{tikzcd}
 \SymBang{\pi}
\arrow[rr,acycof,yshift=.35em]
&&
 \Shuffle_{\tysf}
\arrow[ll,acyclic',fibrant,yshift=-.35em]
\end{tikzcd}
\,.
\end{equation}

The cofibration is defined 
as the functor that maps every word $\bangword{p_1,\dots, p_k} $ to the same word equipped
with the identity deshuffle isomorphism
\[
\begin{tikzcd}
{\identity{\Xmonobjc}}
\quad : \quad
{\tysf (p_1) \otimes \cdots \otimes \tysf (p_k)}
\arrow[rr]
&&
{\tysf (p_1) \otimes \cdots \otimes \tysf (p_k)}
\end{tikzcd}\,
\]
from the object $\Xmonobjc =   \tysf (p_1) \otimes \cdots \otimes \tysf (p_k)$.
The functor is clearly injective on objects, and fully faithful.
It is also essentially surjective: each of a word and a deshuffle isomorphism is isomorphic
to the same word equipped with the identity deshuffle isomorphism.

The fibration is defined as the forgetful functor which
maps a pair $(\bangword{p_1,\dots, p_k},\varphi)$ to the underlying tuple $\bangword{p_1,\dots, p_k}$.
This functor is an isofibration.
Indeed, given a pair $(\bangword{p_1,\dots, p_k},\varphi)$
consisting of a tuple $\bangword{p_1,\dots, p_k} $ with deshuffle componentwise isomorphism
$\varphi : \gamma \to \tysf(p_1) \ \otimes \ \dotsc  \ \otimes \ \tysf (p_k) $,
together with an isomorphism
\(\genbangmorph{\sigma}{h}{k}:\bangword{p_1,\dots, p_k}\to\bangword{p'_1,\dots, p'_k}\)
in \(\SymBang{\pi}\), with \(h_i:p_i\to p'_{\sigma(i)}\) for \(1\le i\le k\),
it is sufficient to consider the object
\(\bangword{p'_1,\dots, p'_k}\) together with the identity deshuffle
with the appropriate definition of~$\bar{h}$.

\paragraph{The fibrations (= isofibrations)}

We now define the separately fibrant span mentioned above in \eqref{equation/commutative-lifting-diagram}
\begin{equation}\label{equation/shufflefibration}
    \begin{tikzcd}[column sep=1.5em]
{\sfiblift{\tysf}}
\quad : \quad \Shuffle_{\tysf}
\arrow[rr,sepfib]
&&
\SymWhyNot{\Gamma}
\end{tikzcd}.
\end{equation}
by the maps
\[   (\bangword{p_1,\dots, p_k},\varphi : \gamma \to \tysf(p_1) \otimes \cdots \otimes \tysf(p_k)) \qquad \qquad \mapsto  \qquad \qquad     \gamma      \]
on objects, and 
\[  \genbangmorph{\sigma}{h}{k}:
    (\bangword{p_1,\dots, p_k},\varphi)
    \to
    (\bangword{p'_1,\dots, p'_k},\varphi')
    \qquad \qquad \mapsto  \qquad \qquad \bar{h} \]
where \(\bar{h}:\gamma\to\gamma'\) is the unique morphism induced 
by \(\genbangmorph{\sigma}{h}{k}\)
as in \eqref{equation/Shuffle}.
Before showing it is separately fibrant, we introduce additional
notations for its components.

Writing \(\tysf = \tuple{f_1,\dotsc,f_n}\) as a tuple of functors with 
\(f_i:\pi\to \SymWhyNot{A_i}\) for \(1\le i\le n\),
we may write \(\varphi = (\varphi_1,\dots,\varphi_n):
\gamma=(\tyl_1,\dotsc,\tyl_n) \to \tysf(p_1) \otimes \cdots \otimes \tysf(p_k)\),
where each morphism 
\(\varphi_i: \tyl_i \to f_i (p_1) \otimes \cdots \otimes f_i (p_k)\)
is a deshuffle isomorphism in \(\SymWhyNot{A_i}\).
Similarly, given \(\genbangmorph{\sigma}{h}{k}\), we may 
write the induced \(\bar{h}=(\alpha_1,\dotsc,\alpha_n):\gamma\to \gamma'\)
with each \(\alpha_i:\tyl_i\to\tyl'_i\) a morphism in \(\SymWhyNot{A_i}\).
Then \(\sfiblift{\tysf}=\tuple{\sfiblift{f_1},\dotsc,\sfiblift{f_n}}\)
where \(\sfiblift{f_i}\) is the functor sending the 
object \((\bangword{p_1,\dots, p_k},\varphi)\)
to \(\tyl_i\) and the morphism \(\genbangmorph{\sigma}{h}{k}\)
to \(\alpha_i\).
Remark that, although we retained the notation \(\sfiblift{f_i}\) for conciseness,
the definition of each component of \(\sfiblift{\tysf}\) depends on the whole
tuple \(\tysf\), not just on the corresponding \(f_i\).

Now, fixing \(1\le i_0\le n\), we show that \(\sfiblift{f_{i_0}}\) is an isofibration.
Let \(\alpha:\tyl\to\tyl'\) be an isomorphism of \(\SymWhyNot{A_{i_0}}\),
and \( (\bangword{p'_1,\dots, p'_k},\varphi') \) an object
of \( \Shuffle_{\tysf} \), mapped to \(\tyl'\) by \(\sfiblift{f_{i_0}}\).
Using the same notations as above, this means \(\tyl'_{i_0}=\tyl'\).
We want to construct an object
\( (\bangword{p_1,\dots, p_k},\varphi) \)
and an isomorphism
\[
    \genbangmorph{\sigma}{h}{k}:
    (\bangword{p_1,\dots, p_k},\varphi)
    \to
    (\bangword{p'_1,\dots, p'_k},\varphi')
\]
which is mapped to \(\alpha\) by \(\sfiblift{f_{i_0}}\):
\ie, we want \eqref{equation/Shuffle} to commute,
each \(h_j\) to be an isomorphism, and
the component \(\alpha_{i_0}\) of \(\bar{h}\) to be \(\alpha\).

By the shuffle factorization system described in \S\ref{section/factorization}, 
the composite isomorphism $\varphi'_{i_0} \circ \alpha$ can be factored as
\begin{equation}
    \label{equation/factorization-sfiblift}
    \varphi'_{i_0} \circ \alpha
    \   =  \ 
    (  \beta_1 \otimes \cdots \otimes \beta_k )\circ \psi
\end{equation}
where 
\begin{center}
\begin{tikzcd}[column sep = 1em]
    \psi \quad : \quad \tyl  \arrow[rr] &&  \tylother_1 \otimes \cdots \otimes \tylother_k
\end{tikzcd}
\end{center}
is a deshuffle isomorphism in \(\SymWhyNot{A_{i_0}}\),
and each $\beta_j : \tylother_j \to f_{i_0} (p'_j)  $
is an isomorphism for \(1\le j\le k\).
At this stage, we use the fact that $f_{i_0}$ is an isofibration,
and deduce that there are objects $ p_j \in \pi $ 
and isomorphisms $  h_j :   p_j \to p'_j $ 
such that $ f_{i_0} (p_j) = \tylother_j$
and $  f_{i_0} (h_j) = \beta_j$
for \(1\le j\le k\).

To complete the definition of the required
isomorphism in \(\Shuffle_{\tysf}\),
we take \(\sigma\) to be the identity permutation,
and we are only left to specify the 
deshuffle isomorphisms \(\varphi_i\)
for \(1\le i\le n\),
in such a way that the commutation of \eqref{equation/Shuffle}
(whose lower row is reduced to \(\varphi'\),
because \(\push{\sigma}\) is the identity)
is compatible with the desired identity:
\(\alpha_{i_0}=\alpha\).
We thus set \(\varphi_{i_0}=\psi\):
the commutation \eqref{equation/factorization-sfiblift}
becomes
\[
    \varphi'_{i_0} \circ \alpha
    \   =  \ 
    (  f_{i_0} (h_1) \otimes \cdots \otimes f_{i_0} (h_k) )\circ \varphi_{i_0}
\]
which ensures the commutation of the \(i_0\)-th component of \eqref{equation/Shuffle},
with \(\alpha_{i_0}=\alpha\).
It remains only to choose deshuffle isomorphisms
\(\varphi_i: \tyl_i \to f_i (p_1) \otimes \cdots \otimes f_i (p_k)\)
in \(\SymWhyNot{A_i}\) for \(i\not=i_0\):
it suffices to take \(\tyl_i=f_i (p_1) \otimes \cdots \otimes f_i (p_k)\)
and \(\varphi_i\) the corresponding identity, which is a deshuffle
isomorphism.
Note that the value of \(\alpha_{i}\) induced by this choice
is irrelevant for our purpose when \(i\not=i_0\):
the existence of an arbitrary \(\varphi_i\) is sufficient.

\section{Conclusion}
\label{section/conclusion}

The homotopy interpretation of linear logic based on spans~\cite{Mellies19lics}
associates to every proof of linear logic
a \emph{homotopy invariant} modulo cut-elimination,
where homotopy is understood
as categorical equivalence between categories in~$\Cat$.
In the present paper, we give a detailed and direct description
of the homotopy interpretation of proofs, where the \emph{fibrant interpretation} 
used in~\cite{Mellies19lics} is replaced by a more liberal \emph{separately fibrant interpretation}
carefully chosen to remain as simple as possible.
The construction relies on the existence of a number of simple and explicit separately fibrant replacements
of the functors (concatenation tensor product, distributivity) underlying the rules of linear logic
(in particular, promotion).
%
%

One purpose -- and in fact the initial motivation -- of this explicit
reconstruction of the original model~\cite{Mellies19lics} is to investigate how the syntactic
presentations~\cite{TAO17,Olimpieri21} of the distributor interpretation of
linear logic~\cite{FioreGambinoHylandWinskel07}, could be adapted to the
homotopy span model.
A decisive step in the process of designing such a syntactic counterpart
is to rely on a presentation as simple and explicit as possible of the homotopy span model.
In that perspective, an important challenge for future work will be 
to understand how our combinatorial study of shuffles could be combined
and unified with the analysis of symmetries developed by Clairambault and Forest,
where an explicit construction of homotopy pullbacks using orthogonality is described
in a variant of the span bicategory~\cite{CF2023,CF2024}.

Also, our interpretation of linear logic based on separately fibrant spans 
has a definite polycategorical flavour. 
We plan to investigate more in detail the axiomatics of the 2-dimensional
polycategorical structure of the interpretation.
A complementary direction for future investigations would be
to work at a tricategorical level
and to define the 2-cells up to an isomorphism (building on~\cite{hoffn:spanmon}), 
instead of strictly as done in~\cite{Mellies19lics} and here.
This alternative definition would require to compose spans 
using homotopy pullbacks instead of usual pullbacks,
with tentative links to the recent work on the groupoid span interpretation
of linear logic~\cite{CF2023} and with the $ \infty$-categorical construction of~\cite{HaringtonMimram25}. 
Another perspective would be to replace the $ \Sym$ monad with 
another monoidal construction on small categories, such as the finite product completion. 
This variation has already been considered in the case of distributors \cite{fiore:cart,Galal20,Olimpieri20}
and it would thus be interesting to see how it works on the homotopy span model.


\bibliography{biblio}

\begin{thebibliography}{10}
\providecommand{\url}[1]{\texttt{#1}}
\providecommand{\urlprefix}{ }
\providecommand{\eprint}[2][]{\url{#2}}

\bibitem{AguiarMahajan10}
Aguiar, M. and S.~Mahajan, \emph{Monoidal Functors, Species and Hopf Algebras},
  volume~29 of \emph{CRM Monograph Series}, American Mathematical Society
  (2010).
\newline\urlprefix\url{https://doi.org/10.1090/crmm/029}

\bibitem{BucciarelliEhrhard00}
Bucciarelli, A. and T.~Ehrhard, \emph{On phase semantics and denotational
  semantics in multiplicative-additive linear logic}, Ann. Pure Appl. Log.
  \textbf{102}, pages 247--282 (2000).
\newline\urlprefix\url{https://doi.org/10.1016/S0168-0072(99)00040-8}

\bibitem{BucciarelliEhrhard01}
Bucciarelli, A. and T.~Ehrhard, \emph{On phase semantics and denotational
  semantics: the exponentials}, Ann. Pure Appl. Log. \textbf{109}, pages
  205--241 (2001).
\newline\urlprefix\url{https://doi.org/10.1016/S0168-0072(00)00056-7}

\bibitem{burroni1971tcat}
Burroni, A., \emph{{$T$-cat\'egories (cat\'egories dans un triple)}}, Cahiers
  de topologie et g\'eom\'etrie diff\'erentielle \textbf{12}, pages 215--321
  (1971).

\bibitem{CF2023}
Clairambault, P. and S.~Forest, \emph{The cartesian closed bicategory of thin
  spans of groupoids}, in: \emph{{LICS}}, pages 1--13 (2023).
\newline\urlprefix\url{https://doi.org/10.1109/LICS56636.2023.10175754}

\bibitem{CF2024}
Clairambault, P. and S.~Forest, \emph{An analysis of symmetry in quantitative
  semantics}, in: P.~Sobocinski, U.~D. Lago and J.~Esparza, editors,
  \emph{{LICS}}, {ACM} (2024).
\newline\urlprefix\url{https://doi.org/10.1145/3661814.3662092}

\bibitem{deCarvalho18}
de~Carvalho, D., \emph{Execution time of {\(\lambda\)}-terms via denotational
  semantics and intersection types}, Math. Struct. Comput. Sci. \textbf{28},
  pages 1169--1203 (2018).
\newline\urlprefix\url{https://doi.org/10.1017/S0960129516000396}

\bibitem{ehrhard-completeness}
Ehrhard, T., \emph{A completeness theorem for symmetric product phase spaces},
  J. Symb. Log. \textbf{69}, pages 340--370 (2004).
\newline\urlprefix\url{https://doi.org/10.2178/JSL/1082418530}

\bibitem{ehr:fin}
Ehrhard, T., \emph{Finiteness spaces}, Math. Struct. Comput. Sci. \textbf{15},
  pages 615--646 (2005).
\newline\urlprefix\url{https://doi.org/10.1017/S0960129504004645}

\bibitem{Ehrhard18}
Ehrhard, T., \emph{An introduction to differential linear logic: proof-nets,
  models and antiderivatives}, Math. Struct. Comput. Sci. \textbf{28}, pages
  995--1060 (2018).
\newline\urlprefix\url{https://doi.org/10.1017/S0960129516000372}

\bibitem{EhrhardRegnier06}
Ehrhard, T. and L.~Regnier, \emph{Differential interaction nets}, Theor.
  Comput. Sci. \textbf{364}, pages 166--195 (2006).
\newline\urlprefix\url{https://doi.org/10.1016/J.TCS.2006.08.003}

\bibitem{FioreGalalPaquet24}
Fiore, M., Z.~Galal and H.~Paquet, \emph{Stabilized profunctors and stable
  species of structures}, Log. Methods Comput. Sci. \textbf{20} (2024).
\newline\urlprefix\url{https://doi.org/10.46298/LMCS-20(1:17)2024}

\bibitem{FioreGambinoHylandWinskel07}
Fiore, M., N.~Gambino, M.~Hyland and G.~Winskel, \emph{{The cartesian closed
  bicategory of generalised species of structures}}, Journal of the London
  Mathematical Society \textbf{77}, pages 203--220 (2007), ISSN 0024-6107.
\newline\urlprefix\url{https://doi.org/10.1112/jlms/jdm096}

\bibitem{fiore:rel}
Fiore, M., N.~Gambino, M.~Hyland and G.~Winskel, \emph{Relative pseudomonads,
  kleisli bicategories, and substitution monoidal structures}, Selecta
  Mathematica  (2017).
\newline\urlprefix\url{https://doi.org/10.1007/s00029-017-0361-3}

\bibitem{fiore:cart}
Fiore, M. and A.~Joyal, \emph{Theory of para-toposes} (2015). Talk at the
  Category Theory Conference.

\bibitem{Galal20}
Galal, Z., \emph{A profunctorial scott semantics}, in: \emph{5th International
  Conference on Formal Structures for Computation and Deduction, {FSCD} 2020,
  June 29-July 6, 2020, Paris, France (Virtual Conference)}, volume 167 of
  \emph{LIPIcs}, pages 16:1--16:18 (2020).
\newline\urlprefix\url{https://doi.org/10.4230/LIPICS.FSCD.2020.16}

\bibitem{joyal:op}
Gambino, N. and A.~Joyal, \emph{On operads, bimodules and analytic functors},
  Memoirs of the American Mathematical Society \textbf{249} (2017), ISSN
  1947-6221.
\newline\urlprefix\url{https://doi.org/10.1090/memo/1184}

\bibitem{Girard87}
Girard, J., \emph{Linear logic}, Theor. Comput. Sci. \textbf{50}, pages 1--102
  (1987).
\newline\urlprefix\url{https://doi.org/10.1016/0304-3975(87)90045-4}

\bibitem{GrelloisMellies15}
Grellois, C. and P.~Melli{\`{e}}s, \emph{Relational semantics of linear logic
  and higher-order model checking}, in: S.~Kreutzer, editor, \emph{24th {EACSL}
  Annual Conference on Computer Science Logic, {CSL} 2015, September 7-10,
  2015, Berlin, Germany}, volume~41 of \emph{LIPIcs}, pages 260--276, Schloss
  Dagstuhl - Leibniz-Zentrum f{\"{u}}r Informatik (2015).
\newline\urlprefix\url{https://doi.org/10.4230/LIPICS.CSL.2015.260}

\bibitem{HaringtonMimram25}
Harington, E. and S.~Mimram, \emph{{\(\infty\)}-categorical models of linear
  logic}, in: M.~Fern{\'{a}}ndez, editor, \emph{10th International Conference
  on Formal Structures for Computation and Deduction, {FSCD} 2025, July 14-20,
  2025, Birmingham, {UK}}, volume 337 of \emph{LIPIcs}, pages 23:1--23:20,
  Schloss Dagstuhl - Leibniz-Zentrum f{\"{u}}r Informatik (2025).
\newline\urlprefix\url{https://doi.org/10.4230/LIPICS.FSCD.2025.23}

\bibitem{Hirschhorn03}
Hirschhorn, P.~S., \emph{Model categories and their localizations}, volume~99
  of \emph{Mathematical Surveys and Monographs}, American Mathematical Society
  (2003).
\newline\urlprefix\url{https://doi.org/10.1090/surv/099}

\bibitem{hoffn:spanmon}
Hoffnung, A.~E., \emph{Spans in 2-categories: A monoidal tricategory} (2013).
  \eprint{1112.0560}.
\newline\urlprefix\url{https://arxiv.org/abs/1112.0560}

\bibitem{Hovey99}
Hovey, M., \emph{Model categories}, volume~63 of \emph{Mathematical Surveys and
  Monographs}, American Mathematical Society (1999).
\newline\urlprefix\url{https://doi.org/10.1090/surv/063}

\bibitem{Hyland10}
Hyland, M., \emph{Some reasons for generalising domain theory}, Mathematical
  Structures in Computer Science \textbf{20}, page 239–265 (2010).
\newline\urlprefix\url{https://doi.org/10.1017/S0960129509990375}

\bibitem{JoyalTierney1991}
Joyal, A. and M.~Tierney, \emph{Strong stacks and classifying spaces}, in:
  A.~Carboni, M.~C. Pedicchio and G.~Rosolini, editors, \emph{Category Theory},
  pages 213--236, Springer Berlin Heidelberg, Berlin, Heidelberg (1991), ISBN
  978-3-540-46435-8.
\newline\urlprefix\url{https://doi.org/10.1007/BFb0084222}

\bibitem{KobayashiOng11}
Kobayashi, N. and C.~L. Ong, \emph{Complexity of model checking recursion
  schemes for fragments of the modal mu-calculus}, Log. Methods Comput. Sci.
  \textbf{7} (2011).
\newline\urlprefix\url{https://doi.org/10.2168/LMCS-7(4:9)2011}

\bibitem{Mellies19lics}
Melliès, P.-A., \emph{Template games and differential linear logic}, in:
  \emph{2019 34th Annual ACM/IEEE Symposium on Logic in Computer Science
  (LICS)}, pages 1--13 (2019).
\newline\urlprefix\url{https://doi.org/10.1109/LICS.2019.8785830}

\bibitem{Olimpieri20}
Olimpieri, F., \emph{Intersection Types and Resource Calculi in the
  Denotational Semantics of Lambda-Calculus}, Ph.D. thesis, Aix-Marseille
  Universit{\'e} (2020).\\
  Available online at \url{https://theses.hal.science/tel-03123485}

\bibitem{Olimpieri21}
Olimpieri, F., \emph{Intersection type distributors}, in: \emph{36th Annual
  {ACM/IEEE} Symposium on Logic in Computer Science, {LICS} 2021, Rome, Italy,
  June 29 - July 2, 2021} (2021).
\newline\urlprefix\url{https://doi.org/10.1109/LICS52264.2021.9470617}

\bibitem{paquet2023strong}
Paquet, H. and P.~Saville, \emph{Effectful semantics in bicategories: strong,
  commutative, and concurrent pseudomonads} (2024).
\newline\urlprefix\url{https://doi.org/10.1145/3661814.3662130}

\bibitem{Rezk1996}
Rezk, C., \emph{A model category for categories} (1996). Unpublished note.

\bibitem{Terui12}
Terui, K., \emph{Semantic evaluation, intersection types and complexity of
  simply typed lambda calculus}, in: A.~Tiwari, editor, \emph{23rd
  International Conference on Rewriting Techniques and Applications (RTA'12) ,
  {RTA} 2012, May 28 - June 2, 2012, Nagoya, Japan}, volume~15 of
  \emph{LIPIcs}, pages 323--338, Schloss Dagstuhl - Leibniz-Zentrum f{\"{u}}r
  Informatik (2012).
\newline\urlprefix\url{https://doi.org/10.4230/LIPICS.RTA.2012.323}

\bibitem{TAO17}
Tsukada, T., K.~Asada and C.~L. Ong, \emph{Generalised species of rigid
  resource terms}, in: \emph{32nd Annual {ACM/IEEE} Symposium on Logic in
  Computer Science, {LICS} 2017, Reykjavik, Iceland, June 20-23, 2017}, pages
  1--12, {IEEE} Computer Society (2017).
\newline\urlprefix\url{https://doi.org/10.1109/LICS.2017.8005093}

\end{thebibliography}

\end{document}